\documentclass[reqno]{amsart}
\usepackage{amsmath,amsthm,amssymb,mathrsfs,stmaryrd,eucal,mathbbol}
\usepackage[all,cmtip]{xy}
\usepackage{hyperref}
\usepackage{tikz}

\numberwithin{equation}{section}

\newcommand\void[1]{}

\newcommand{\C}{\mathbb{C}}

\newcommand{\R}{\mathbb{R}}

\newcommand{\CA}{\mathcal{A}}
\newcommand{\CB}{\mathcal{B}}
\newcommand{\CC}{\mathcal{C}}
\newcommand{\CD}{\mathcal{D}}
\newcommand{\CE}{\mathcal{E}}
\newcommand{\CF}{\mathcal{F}}
\newcommand{\CG}{\mathcal{G}}
\newcommand{\CH}{\mathcal{H}}

\newcommand{\CK}{\mathcal{K}}
\newcommand{\CL}{\mathcal{L}}

\newcommand{\CO}{\mathcal{O}}
\newcommand{\CP}{\mathcal{P}}
\newcommand{\CQ}{\mathcal{Q}}
\newcommand{\CR}{\mathcal{R}}
\newcommand{\CS}{\mathcal{S}}
\newcommand{\CT}{\mathcal{T}}

\newcommand{\CX}{\mathcal{X}}

\newcommand{\FZ}{\mathfrak{Z}}

\newcommand{\Hilb}{\mathrm{Hilb}}

\newcommand{\op}{\mathrm{op}}

\newcommand{\TO}{\mathcal{QL}}
\newcommand{\TOsk}{\mathcal{QL}_\mathrm{sk}}

\newcommand{\Net}{\mathcal{N}et}

\newcommand{\LQS}{\mathcal{LQS}}
\newcommand{\sk}{\mathrm{sk}}
\newcommand{\lqs}{\mathrm{lqs}}

 \DeclareMathOperator{\Hom}{Hom}
 \DeclareMathOperator{\End}{End}

 \DeclareMathOperator{\Id}{Id}

 \DeclareMathOperator{\one}{\mathbf1}

 \DeclareMathOperator{\Fun}{Fun}

 \DeclareMathOperator{\LMod}{LMod}
 
 \DeclareMathOperator{\BMod}{BMod}
 \DeclareMathOperator{\Diff}{Diff}
 \DeclareMathOperator{\Conf}{Conf}
 
 \DeclareMathOperator{\Isom}{Isom}

 \DeclareMathOperator{\Rep}{Rep}
 \DeclareMathOperator{\Disk}{Disk}
 \DeclareMathOperator{\VN}{VN}
 \DeclareMathOperator{\Ad}{Ad}
 \DeclareMathOperator{\Mor}{Mor}
 \DeclareMathOperator{\Sect}{Sect}

\newtheorem{thm}{Theorem}[section]
\newtheorem{lem}[thm]{Lemma}
\newtheorem{prop}[thm]{Proposition}
\newtheorem{cor}[thm]{Corollary}

\theoremstyle{definition}
\newtheorem{defn}[thm]{Definition}

\newtheorem{exam}[thm]{Example}
\newtheorem{rem}[thm]{Remark}

\theoremstyle{remark}

\newcommand\arXiv[1]{\href{http://arxiv.org/abs/#1}{\nolinkurl{arXiv:#1}}}

\newcommand\condense{\mathrel{\,\hspace{.75ex}\joinrel\rhook\joinrel\hspace{-.75ex}\joinrel\rightarrow}}

\begin{document}

\title{Categories of quantum liquids III}
\maketitle

\begin{center}
Liang Kong$^{a,b,c}$,
Hao Zheng$^{d,e}$
~\footnote{Emails:{\tt  kongl@sustech.edu.cn, haozheng@mail.tsinghua.edu.cn}}
\\[1em]
{\small $^a$ Shenzhen Institute for Quantum Science and Engineering, \\
Southern University of Science and Technology, Shenzhen, 518055, China 
\\[0.2em]
$^b$ International Quantum Academy, 
%and Shenzhen Branch, Hefei National Laboratory, Futian District, Shenzhen, China
Shenzhen 518048, China
%Southern University of Science and Technology, Shenzhen, 518055, China
\\[0.2em]
$^c$ Guangdong Provincial Key Laboratory of Quantum Science and Engineering, \\
Southern University of Science and Technology, Shenzhen, 518055, China
\\[0.2em]
$^d$ Institute for Applied Mathematics, Tsinghua University, Beijing, 100084, China
\\[0.2em]
$^e$ Beijing Institute of Mathematical Sciences and Applications, 
Beijing 101408, China }
%\\[0.4em] $^g$ Department of Mathematics, Peking University, Beijing 100871, China
\end{center}

\smallskip
\begin{abstract}
We continue our study of the categories of quantum liquids started in a previous work. We combine local quantum symmetries with topological skeletons into a single mathematical theory of topological nets and defect nets. In particular, we introduce the notion of a topological net, which is motivated from and generalizes that of a conformal net, and the notion of a defect net which generalizes that of a defect between conformal nets. We give explicit examples of them. Moreover, we construct the category of topological $n$-nets with $k$-morphisms defined by defect $n$-nets of codimension $k$, and show that the category of $n$D quantum liquids can be extracted from it and computed explicitly via the condensation theory of topological nets. 
\end{abstract}

\tableofcontents

\section{Introduction}

In this work, we continue our study of the category of quantum liquids started in \cite{KZ20b}. Quantum liquids are quantum phases that only ``softly'' depend on the local geometry on spacetime. They include the usual spontaneous symmetry-breaking phases, topological orders, symmetry protected/enriched topological (SPT/SET) orders and CFT-type gapless phases. We have shown in \cite{KZ20b} that the mathematical description of a quantum liquid consists of two parts of data: the local quantum symmetry and the topological skeleton. In \cite{KZ20b}, we focused on the topological skeletons and explicitly computed the category $\TOsk^n$ of the topological skeletons of $n$D  quantum liquids. In this work, we combine local quantum symmetries with  topological skeletons into a single mathematical theory of topological nets and defect nets. This theory provides a rather complete mathematical description of quantum liquids in all dimensions. Throughout this work, $n$D represents the spacetime dimension. 

\smallskip
In 2D CFT's, the local quantum symmetries are given by vertex operator algebras (VOA) or their non-chiral analogues. It is not known how to generalize VOA directly to higher dimensions. However, there is an alternative formulation of 2D CFT's in terms of conformal nets (see for examples \cite{BGL93,BMT88,BSM90,GF93,KL04,KLM01,LR95,Reh00a,Reh00b,Was95} and references therein). The idea of conformal net was originated from algebraic quantum field theories (defined in all dimensions), which was based on Haag-Kastler nets or the nets of local observables \cite{HK64} (see \cite{Haa92} for a review).  Therefore, it is natural to ask how to generalize the notion of a conformal net to a more general one for the study of quantum liquids in all dimensions. Moreover, from the lattice model realization of topological orders, it seems that the ideas of the net of local observables and the superselection sectors of particles work better in lattice models than in quantum gauge field theories (see for example \cite{Kit03,KK12,KWZ21}). This motivates us to introduce the notion of a topological net, which includes that of a conformal net as a special case. 
We also generalize Bartels, Douglas and Henriques' theory of defects\footnote{Certain defects in conformal nets were studied much earlier than \cite{BDH19a} under the name of `solitons' (see for example \cite{BE98,Kaw02,LR95,LX04}).} in conformal nets and their fusion \cite{BDH19a} to the theory of defect nets and their fusion.

In this work, we show that all finite topological $n$-nets, together with their higher codimensional finite defects, form a symmetric monoidal $*$-$(n+1)$-category $\Net^n$. Then we can define a subcategory $\LQS^n$ of $\Net^n$. The category of $n$D quantum liquids, denoted by $\TO^n$, can be obtained from $\LQS^n$ by a (co)slice construction. The category $\TO^n$ is equipped with two forgetful functors $\TO^n \to \LQS^n$ and $\TO^n \to \TOsk$, whose images are precisely the local quantum symmetry and the topological skeleton, respectively. One of the main results of this work says that $\LQS^n$ is $*$-condensation-complete (see Theorem \ref{thm:lqs-cc}). It further implies $\LQS^n\simeq (n+1)\Hilb$ and $\TO^n\simeq\TOsk^n$, both of which are naturally required by physics thus provide the necessary consistent check of our theory.

\smallskip
We provide the layout of this paper. 
In Section \ref{sec:topological-net}, we introduce the notion of a topological net; in Section \ref{sec:defect-net}, we introduce that of a defect net; in Section \ref{sec:net-oss}, we construct from every finite group a topological $n$-net that describe the local quantum symmetry of $n$D gapped quantum liquids with a finite onsite symmetry; in Section \ref{sec:lw-net}, we provide an explicit construction of topological nets and defect nets to describe Levin-Wen models and their boundaries.

In Section \ref{sec:netn}, we introduce the category $\Net^n$ and postpone the complete construction (mainly the composition of higher morphisms) of $\Net^n$ to Section \ref{sec:fus-def}. In Section \ref{sec:construct-QL}, we give the construction of $\LQS^n$ and $\TO^n$, state the main result Theorem \ref{thm:lqs-cc} and discuss its consequences. In Section \ref{sec:transparent-wall}, we show how to extract the information of local quantum symmetries from a tower of subcategories of $\LQS^n$ based on the so-called transparent domain walls. 

Section \ref{sec:fus-def} is devoted to defining the fusion of defect nets and completing the construction of $\Net^n$. In Section \ref{sec:condense-net}, we sketch the condensation theory of topological nets. This theory provides a proof of Theorem \ref{thm:lqs-cc}. 

In Appendix \ref{sec:condense-vn}, we briefly review Lurie's formulation of Connes fusion and prove the $*$-condensation completeness of the $*$-2-category of von Neumann algebras and bimodules.

\medskip
We assume the reader is familiar with the fundamentals of von Neumann algebras, such as standard form and Connes fusion of bimodules. We say that a bimodule (including a left or right module as special case) over von Neumann algebras is semisimple if it is a finite direct sums of irreducible ones.

We work on the $*$-settings. Functors between $*$-$n$-categories are silently assumed to be $*$-functors.
Recall that $\Hilb$ denotes the unitary symmetric monoidal 1-category of finite-dimensional Hilbert spaces and linear maps. We use $\widehat\Hilb$ to the denote the symmetric monoidal $*$-1-category of all (possibly inseparable) Hilbert spaces and bounded linear maps. 
Given a unitary 1-category $\CC$, we use $\hat\CC$ to denote $\CC\boxtimes_\Hilb\widehat\Hilb$ and refer to it as the {\em completion} of $\CC$. By definition $\hat\CC$ is a finite direct sum of $\widehat\Hilb$.

\medskip
\noindent{\bf Acknowledgments}:
HZ would like to thank Zhengwei Liu for helpful discussions on von Neumann algebras. LK is supported by NSFC under Grant No. 11971219, and by Guangdong Provincial Key Laboratory (Grant No.2019B121203002) and by Guangdong Basic and Applied Basic Research Foundation under Grant No. 2020B1515120100. HZ is supported by NSFC under Grant No. 11871078.

\section{Topological nets}

The notion of a topological net is a generalization of a conformal net. By dropping the continuity of diffeomorphism covariance, we obtain a theory unifying conformal symmetries of gapless phases and onsite symmetries of gapped phases.

\subsection{Topological nets} \label{sec:topological-net}

We adopt the standard definition of the $n$-sphere
$$S^n=\{ (x_0,x_1,\dots,x_n) \mid \sum x_i^2=1 \}.$$
Let $S^n_\uparrow$ denote the upper half of $S^n$ defined by $x_0\ge0$ and let $S^n_\downarrow$ denote the lower half of $S^n$ defined by $x_0\le0$.

A {\em disk region} of the $n$-sphere $S^n$ is a closed region that is diffeomorphic to the $n$-disk $D^n$.
Let $\Disk^n$ denote the 1-category of disk regions of $S^n$ whose morphisms are inclusions of disk regions. Let $\VN$ denote the 1-category of von Neumann algebras whose morphisms are (unital, normal) homomorphisms of von Neumann algebras. 

\begin{defn}
An {\em $n$-dimensional net of von Neumann algebras} or an {\em $n$-net} for short is a functor $\CA:\Disk^{n-1}\to\VN$. A {\em partial $n$-net} is a functor $\CA:\CD\to\VN$ where $\CD$ is a full subcategory of $\Disk^{n-1}$.
\end{defn}

\begin{rem}
Unwinding the definition we see that a partial $n$-net $\CA:\CD\to\VN$ consists of a family of von Neumann algebras $\CA(I)$ for $I\in\CD$ and a family of homomorphisms $\CA(I\subset J):\CA(I)\to\CA(J)$ such that $\CA(I\subset I) = \Id_{\CA(I)}$ and $\CA(J\subset K)\circ\CA(I\subset J) = \CA(I\subset K)$. The von Neumann algebras $\CA(I)$ are referred to as {\em local observable algebras}.
\end{rem}

\begin{rem}
We do not require the isotony of a partial $n$-net, that is, $\CA(I)$ may be not a subalgebra of $\CA(J)$ for $I\subset J$. The isotony of a topological $n$-net follows from other axioms (see Remark \ref{rem:tn}(2)). This relaxation is necessary for defining a zero defect $n$-net (see Example \ref{exam:tn2dn}).
\end{rem}

The {\em direct sum} of two $n$-nets $\CA,\CB$ is an $n$-net $\CA\oplus\CB$ defined by
$$(\CA\oplus\CB)(I) = \CA(I)\oplus\CB(I).$$
The {\em tensor product} of $\CA,\CB$ is an $n$-net $\CA\boxtimes\CB$ defined by
$$(\CA\boxtimes\CB)(I) = \CA(I)\bar\otimes\CB(I).$$

The diffeomorphism group $\Diff(S^{n-1})$ acts on the collection of partial $n$-nets by the formula: $(h^*\CA)(I) = \CA(h(I))$ if $h$ preserves orientation or $(h^*\CA)(I) = \CA(h(I))^\op$ otherwise.

\begin{defn} \label{defn:sector}
A {\em sector} of a partial $n$-net $\CA:\CD\to\VN$ is a Hilbert space $\CH$ equipped with a family of homomorphisms $\rho_I:\CA(I)\to\CL(\CH)$ for $I\in\CD$ such that $\rho_J\circ\CA(I\subset J) = \rho_I$ for all inclusions $I\subset J$.
We say that $\CH$ is {\em irreducible} if it is neither zero nor the direct sum of two nonzero sectors; {\em semisimple} if it is a finite direct sum of irreducible sectors.

A {\em homomorphism} $f:\CH\to\CK$ between two sectors of $\CA$ is a bounded linear map such that $f$ is a left $\CA(I)$-module map for all $I\in\CD$. We use $\Sect(\CA)$ to denote the $*$-1-category of sectors of $\CA$.
\end{defn}

Given a sector $\CH$ of a partial $n$-net $\CA:\CD\to\VN$ and a closed region $R\subset S^{n-1}$ with smooth boundary, we use $\CO_\CA(\CH,R)$ or $\CO(\CH,R)$ to denote the von Neumann algebra on $\CH$ generated by the images of $\CA(I)$ for all $R\supset I\in\CD$. In particular, we denote $\CO_\CA(\CH,S^{n-1})$ by $\CO_\CA(\CH)$ or $\CO(\CH)$.

\begin{rem}
%In the above definition, let $\CO(\CH)\subset\CL(\CH)$ be the von Neumann algebra generated by all of $\rho_I(\CA(I))$. Note that 
A sector $\CH$ of a partial $n$-net is irreducible if and only if $\CH$ is an irreducible left $\CO(\CH)$-module, i.e. $\CO(\CH)=\CL(\CH)$.
Similarly, $\CH$ is semisimple if and only if $\CH$ is a semisimple left $\CO(\CH)$-module. If $\CH$ is semisimple, $\CO(\CH)$ is a finite direct sum of type I factors.
\end{rem}

\begin{defn} \label{defn:top-net}
A {\em topological $n$-net} consists of the following data:
\begin{itemize}
\item an $n$-net $\CA:\Disk^{n-1}\to\VN$;
\item a sector $\CH_\CA$ of $\CA$, called the {\em vacuum sector};
\item a natural isomorphism $\eta_h:\CA\to h^*\CA$ for every $h\in\Diff(S^{n-1})$ such that $h^*(\eta_g) \circ \eta_h = \eta_{g\circ h}$.
\end{itemize}
These data are subject to the following axioms where $I,J,K\in\Disk^{n-1}$:
\begin{itemize}
\item {\em Locality:} If $I,J\subset K$ have disjoint interiors, then the images of $\CA(I)$ and $\CA(J)$ are commuting subalgebras of $\CA(K)$.

\item {\em Additivity:} If $I,J\subset K$ and if the interiors of $I$ and $J$ in $K$ cover $K$, then $\CA(K)$ is generated by the images of $\CA(I)$ and $\CA(J)$.
%If $K=I\cup J$ and if there is a homeomorphism $h:K\to D^{n-2}\times[0,3]$ such that $h(I)=D^{n-2}\times[0,2]$ and $h(J)=D^{n-2}\times[1,3]$,
%If the interior of $K$ is the union of the interiors of $I$ and $J$ and if $\partial K$ is the union of the interiors of $\partial I\cap \partial K$ and $\partial J\cap \partial K$ in $\partial K$,

%\item {\em Split property:} If $I,J\subset K$ are disjoint, then the homomorphism from the algebraic tensor product $\CA(I)\otimes\CA(J) \to \CA(K)$ extends to a homomorphism from the spatial tensor product $\CA(I)\bar\otimes\CA(J) \to \CA(K)$.

\item {\em Covariance:} For every orientation-preserving (resp. orientation-reversing) diffeomorphism $h\in\Diff(S^{n-1})$ there exists an isometry $\alpha_h:\CH_\CA\to\CH_\CA$ (resp. $\alpha_h:\CH_\CA\to\bar\CH_\CA$) rendering the following diagram commutative:
$$\xymatrix{
  \CA(I) \ar[r]^{\eta_{h,I}} \ar[d]_{\rho_I} & \CA(h(I)) \ar[d]^{\rho_{h(I)}} \\
  \CL(\CH_\CA) \ar[r]^{\Ad(\alpha_h)} & \CL(\CH_\CA) 
} \quad\quad \raisebox{-2em}{\text{resp.}} \quad\quad 
\xymatrix{
  \CA(I) \ar[r]^{\eta_{h,I}} \ar[d]_{\rho_I} & \CA(h(I))^\op \ar[d]^{\rho_{h(I)}} \\
  \CL(\CH_\CA) \ar[r]^{\Ad(\alpha_h)} & \CL(\bar\CH_\CA) .
}
$$
Moreover, $\eta_{h,I}=\Id_{\CA(I)}$ if $h|_I=\Id_I$.

\item {\em Vacuum property:} Identify $\CA(S^{n-1}_\downarrow)$ with $\CA(S^{n-1}_\uparrow)^\op$ via $\eta_r$ where $r$ is the reflection across the hyperplane $x_0=0$. Then the $\CA(S^{n-1}_\uparrow)$-$\CA(S^{n-1}_\downarrow)^\op$-bimodule $\CH_\CA$ is isometric to $L^2(\CA(S^{n-1}_\uparrow))$.
\end{itemize}
We use the notations $\CA$ and $(\CA,\CH_\CA)$ interchangeably for a topological $n$-net.
We say that $(\CA,\CH_\CA)$ is {\em irreducible} (resp. {\em semisimple}) if the vacuum sector $\CH_\CA$ is irreducible (resp. semisimple). 
\end{defn}

\begin{exam}
(1) The {\em zero topological $n$-net} $(\underline{0},0)$ where $\underline{0}(I)=0$.

(2) The {\em trivial topological $n$-net} $(\underline{\C},\C)$ where $\underline{\C}(I)=\C$.
\end{exam}

\begin{exam}
If $(\CA,\CH_\CA)$ and $(\CB,\CH_\CB)$ are topological $n$-nets then $(\CA\oplus\CB,\CH_\CA\oplus\CH_\CB)$ and $(\CA\boxtimes\CB,\CH_\CA\otimes\CH_\CB)$ are also topological $n$-nets.
\end{exam}

\begin{rem} \label{rem:tn}
There are several remarks concerning Definition \ref{defn:top-net}:

(1) All of $\CA(I)$ are von Neumann algebras on $\CH_\CA$ and Haag duality holds on $\CH_\CA$:
$$\CA(I') = \CA(I)'$$
where $I'$ is the closure of the complement of $I$.
Indeed, the claim holds for $I=S^{n-1}_\uparrow$ by the vacuum property and holds for general $I$ by the covariance axiom.

(2) As an immediate consequence, $\CA(I)$ is a subalgebra of $\CA(J)$ if $I\subset J$. 

(3) By the last statement of the covariance axiom, the isomorphism $\eta_{h,I}:\CA(I)\to\CA(h(I))$ or $\eta_{h,I}:\CA(I)\to\CA(h(I))^\op$ depends only on $h|_I$. In another word, a diffeomorphism $f:I\to J$ determines an isomorphism $f_*:\CA(I)\to\CA(J)$ or $f_*:\CA(I)\to\CA(J)^\op$.
Moreover, if $h_{I'}=\Id_{I'}$ then $\alpha_h\in\CA(I')'=\CA(I)$ so that $\eta_{h,I}=\Ad(\alpha_h):\CA(I)\to\CA(I)$ is an inner automorphism.

(4) By Haag duality, $\CO(\CH_A)'\subset\CO(\CH_A)$. Hence $\CO(\CH_A)'$ is the center $Z(\CO(\CH_\CA))$ and therefore $\CO(\CH_A)$ is of type I. In particular, $(\CA,\CH_\CA)$ is irreducible if and only if $\CO(\CH_\CA)=\CL(\CH_\CA)$ and if and only if $Z(\CO(\CH_\CA))=\C$. In general, if $(\CA,\CH_\CA)$ is semisimple then it is a finite direct sum of irreducible ones where the decomposition is induced by that of $Z(\CO(\CH_\CA))$.

(5) The isometry $\alpha_h$ is unique up to a unitary in $Z(\CO(\CH_\CA))$. Therefore, if $(\CA,\CH_\CA)$ is irreducible then $\CH_A$ carries a (possibly discontinuous) projective action of $\Diff(S^{n-1})$.
\end{rem}

\begin{defn}
We say that a topological $n$-net $(\CA,\CH_\CA)$ is {\em finite} if it is semisimple and satisfies the following conditions for any (equivalently, some) disjoint $I,J\in\Disk^{n-1}$:
\begin{itemize}
\item {\em Split property:} The homomorphism from the algebraic tensor product $\CA(I)\otimes_{alg}\CA(J) \to \CL(\CH_\CA)$ extends to a homomorphism from the spatial tensor product $\CA(I)\bar\otimes\CA(J) \to \CL(\CH_\CA)$.
\item {\em Duality:} The $\CO(\CH_\CA,I\cup J)$-$\CO(\CH_\CA,I'\cap J')^\op$-bimodule $\CH_\CA$ is dualizable\footnote{We do not required the normalization condition of \cite{BDH14}. See Remark \ref{rem:vnbim-dual}.}.
\end{itemize}
\end{defn}

\begin{exam} \label{exam:1tn}
Note that a 1-net $\CA$ consists of a pair of von Neumann algebras $\CA(1)$ and $\CA(-1)$. For a topological 1-net $(\CA,\CH_\CA)$, $\CA(-1)$ is determined by $\CA(1)$ and $\CH_\CA$ is determined up to isometry. Therefore, a sector of $\CA$ is simply an $\CA(1)$-$\CA(1)$-bimodule.
Moreover, $(\CA,\CH_\CA)$ satisfies the split property if and only if $\CA(1)$ is of type I; $(\CA,\CH_\CA)$ is finite if and only if $\CA(1)$ is a finite direct sum of type I factors.
\end{exam}

\begin{exam} \label{exam:conf-net}
A (finite semisimple) conformal net in the sense of \cite{BDH15,BDH19b} is equivalent to a (finite) topological 2-net satisfying the additional assumptions of continuity, strong additivity and split property. See the terminology therein.
(An irreducible conformal net in the sense of \cite{BDH15} with positive energy spectrum is equivalent to a conformal net in the sense of \cite{GF93} satisfying the additional assumptions of strong additivity and diffeomorphism covariance.)
\end{exam}

\subsection{Defect nets} \label{sec:defect-net}

We identify $S^{n-1}$ with the equator of $S^n$ defined by $x_n=0$ so that $S^n_\uparrow\cap S^{n-1}=S^{n-1}_\uparrow$ and $S^n_\downarrow\cap S^{n-1}=S^{n-1}_\downarrow$.

\smallskip

Let $M$ be a smooth manifold and $N\subset M$ a closed submanifold of codimension one. A {\em local parametrization} of $M$ around $N$ is an equivalence class of smooth embeddings $f:N\times(-1,1)\to M$ that satisfy $f(x,0)=x$ for $x\in N$; two smooth embeddings are equivalent if they agree on a neighborhood of $N\times\{0\}$. For example, there is a standard local parametrization of $S^n$ around $S^{n-1}$ represented by the smooth embedding $(x,s)\mapsto(x\sqrt{1-s^2},s)$.

Let $\Diff(S^n)_k \subset \Diff(S^n)$ be the subgroup of diffeomorphisms $h$ such that $h(S^{n-i})=S^{n-i}$ and $h$ preserves the standard local parametrization of $S^{n-i+1}$ around $S^{n-i}$ for $1\le i\le k$.
Let $\Disk^n_k$ denote the full subcategory of $\Disk^n$ consisting of those disk regions $I$ such that
%, for $1\le i\le k$, if $I$ intersects $S^{n-i}$ then $I=h(S^n_\uparrow)$ for some $h\in\Diff(S^n)_i$.
$I=h(S^n_\uparrow)$ for some $h\in\Diff(S^n)_{d(I)}$ where $d(I)=\max\{i\mid 0\le i\le k, \, I\cap S^{n-i}\ne\emptyset\}$.
%$d$ is the maximal integer satisfying $d\le k$ and $I\cap S^{n-d}\ne\emptyset$.

\begin{defn}
A {\em defect $n$-net} of codimension $k$ where $0\le k<n$ consists of the following data:
\begin{itemize}
\item a partial $n$-net $\CA:\Disk^{n-1}_k\to\VN$;
\item a sector $\CH_\CA$ of $\CA$, called the {\em vacuum sector};
\item a natural isomorphism $\eta_h:\CA\to h^*\CA$ for every $h\in\Diff(S^{n-1})_k$ such that $h^*(\eta_g) \circ \eta_h = \eta_{g\circ h}$.
\end{itemize}
These data are subject to the following axioms where $I,J,K\in\Disk^{n-1}_k$:
\begin{itemize}
\item {\em Locality:} If $I,J\subset K$ have disjoint interiors, then the images of $\CA(I)$ and $\CA(J)$ are commuting subalgebras of $\CA(K)$.

\item {\em Additivity:} If $I,J\subset K$ and if the interiors of $I$ and $J$ in $K$ cover $K$, then $\CA(K)$ is generated by the images of $\CA(I)$ and $\CA(J)$.

%\item {\em Split property:} If $I,J\subset K$ are disjoint, then the homomorphism from the algebraic tensor product $\CA(I)\otimes\CA(J) \to \CA(K)$ extends to a homomorphism from the spatial tensor product $\CA(I)\bar\otimes\CA(J) \to \CA(K)$.

\item {\em Covariance:} For every orientation-preserving (resp. orientation-reversing) diffeomorphism $h\in\Diff(S^{n-1})_k$ there exists an isometry $\alpha_h:\CH_\CA\to\CH_\CA$ (resp. $\alpha_h:\CH_\CA\to\bar\CH_\CA$) rendering the following diagram commutative:
$$\xymatrix{
  \CA(I) \ar[r]^{\eta_{h,I}} \ar[d]_{\rho_I} & \CA(h(I)) \ar[d]^{\rho_{h(I)}} \\
  \CL(\CH_\CA) \ar[r]^{\Ad(\alpha_h)} & \CL(\CH_\CA) 
} \quad\quad \raisebox{-2em}{\text{resp.}} \quad\quad 
\xymatrix{
  \CA(I) \ar[r]^{\eta_{h,I}} \ar[d]_{\rho_I} & \CA(h(I))^\op \ar[d]^{\rho_{h(I)}} \\
  \CL(\CH_\CA) \ar[r]^{\Ad(\alpha_h)} & \CL(\bar\CH_\CA) .
}
$$
Moreover, $\eta_{h,I}=\Id_{\CA(I)}$ if $h|_I=\Id_I$.

\item {\em Vacuum property:} Identify $\CA(S^{n-1}_\downarrow)$ with $\CA(S^{n-1}_\uparrow)^\op$ via $\eta_r$ where $r$ is the reflection across the hyperplane $x_0=0$. Then the $\CA(S^{n-1}_\uparrow)$-$\CA(S^{n-1}_\downarrow)^\op$-bimodule $\CH_\CA$ is isometric to $L^2(\CA(S^{n-1}_\uparrow))$.
\end{itemize}
We use the notations $\CA$ and $(\CA,\CH_\CA)$ interchangeably for a topological $n$-net.
We say that $(\CA,\CH_\CA)$ is {\em irreducible} (resp. {\em semisimple}) if the vacuum sector $\CH_\CA$ is irreducible (resp. semisimple). 
\end{defn}

\begin{rem} \label{rem:res-k-k+1}
A topological $n$-net is exactly a defect $n$-net of codimension zero. A defect $n$-net of codimension $k$ restricts to a defect $n$-net of codimension $k+1$.
\end{rem}

\begin{rem}
(1) By the vacuum property and the covariance axiom, if $I$ intersects $S^{n-1-k}$ then $\CA(I)$ is a von Neumann algebra on $\CH_\CA$ and Haag duality holds on $\CH_\CA$:
$$\CA(I') = \CA(I)'.$$

(2) As an immediate consequence, $\CA(I)$ is a subalgebra of $\CA(J)$ if $I\subset J$ and if $I,J$ intersect $S^{n-1-k}$. 

(3) By the last statement of the covariance axiom, the isomorphism $\eta_{h,I}:\CA(I)\to\CA(h(I))$ or $\eta_{h,I}:\CA(I)\to\CA(h(I))^\op$ depends only on $h|_I$. 
If $I$ intersects $S^{n-1-k}$ and if $h_{I'}=\Id_{I'}$ then $\eta_{h,I}:\CA(I)\to\CA(I)$ is an inner automorphism.

(4) By Haag duality, $\CO(\CH_A)'\subset\CO(\CH_A)$. Hence $\CO(\CH_A)'$ is the center $Z(\CO(\CH_\CA))$ and therefore $\CO(\CH_A)$ is of type I. In particular, $(\CA,\CH_\CA)$ is irreducible if and only if $\CO(\CH_\CA)=\CL(\CH_\CA)$ and if and only if $Z(\CO(\CH_\CA))=\C$. 

(5) The isometry $\alpha_h$ is unique up to a unitary in $Z(\CO(\CH_\CA))$. Therefore, if $(\CA,\CH_\CA)$ is irreducible then $\CH_A$ carries a (possibly discontinuous) projective action of $\Diff(S^{n-1})_k$.
\end{rem}

\begin{defn}
We say that a defect $n$-net $(\CA,\CH_\CA)$ of codimension $k$ is {\em finite} if it is semisimple and satisfies the following conditions for any (equivalently, some) disjoint $I,J\in\Disk^{n-1}_k$ that intersect $S^{n-1-k}$:
\begin{itemize}
\item {\em Split property:} The homomorphism from the algebraic tensor product $\CA(I)\otimes_{alg}\CA(J) \to \CL(\CH_\CA)$ extends to a homomorphism from the spatial tensor product $\CA(I)\bar\otimes\CA(J) \to \CL(\CH_\CA)$.
\item {\em Duality:} The $\CO(\CH_\CA,I\cup J)$-$\CO(\CH_\CA,I'\cap J')^\op$-bimodule $\CH_\CA$ is dualizable.
\end{itemize}
\end{defn}

\begin{exam}
If $(\CA,\CH_\CA)$ is a defect $n$-net of codimension $k$, then $(\CA^\op,\bar\CH_\CA)$ is also a defect $n$-net of codimension $k$ where $\CA^\op(I)=\CA(I)^\op$.
\end{exam}

\begin{exam} \label{exam:tn2dn}
A defect $n$-net $(\CA,\CH_\CA)$ of codimension $k$ induces a defect $(n+1)$-net $(\hat\CA,\CH_\CA)$ of codimension $k+1$ where $\hat\CA(I) = \CA(I\cap S^{n-1})$ if $I$ intersects $S^{n-1}$ or $\hat\CA(I)=\C$ otherwise. In particular, the zero topological $n$-net $\underline{0}$ induces a zero defect $(n+1)$-net of codimension one which do not satisfy the isotony axiom.
\end{exam}

\begin{exam} \label{exam:conf-defect}
A (finite semisimple) defect defined in \cite{BDH19b} is a (finite) defect 2-net of codimension one.
Indeed, a finite semisimple defect satisfies the duality condition by \cite[Proposition 3.18]{BDH19a}.
\end{exam}

\subsection{Nets of onsite symmetries} \label{sec:net-oss}

We define for every finite group $G$ a finite irreducible topological $n$-net that describes the local quantum symmetry of an $n$D gapped phase with an onsite symmetry $G$.

\smallskip

Let $G$ be a finite group and let $V$ be a finite-dimensional $G$-module which contains all the irreducible $G$-modules as direct summands. Fix a $G$-invariant vector of norm one $\mu\in V$.

Consider the orthonormal frame bundle $E$ over $S^{n-1}$, i.e. the fiber $E_x$ over a point $x\in S^{n-1}$ consists of the orthonormal frames at $x$. For example, $E$ is a double cover of $S^1$ for $n=2$. In general, $E$ has two connected components $E^\pm$ corresponding to the two orientations of $S^{n-1}$ (in fact, $E\cong O(n)$).

Assign $W_s=V^*$ for $s\in E^+$ and $W_s=V$ for $s\in E^-$. 
Define $\CH = \bigotimes_{s\in E}W_s$ to be the completion of the pre-Hilbert space spanned by the following vectors 
$$\{ \otimes w_s \mid \text{$w_s=\mu$ or $\mu^*$ except for finitely many $s\in E$} \}.$$
Define $\CH_\CA$ to be the subspace of $G$-invariants in $\CH$
$$\CH_\CA = \CH^G = \{ w\in\CH \mid G w=w \}.$$
It carries an obvious action by the diffeomorphism group $\Diff(S^{n-1})$.

Let $I\subset S^{n-1}$ be a disk region. Note that $I$ and $I'$ divide $E$ into two disjoint parts $E(I)$ and $E(I')$ such that an orthonormal frame $(e_1,\dots,e_{n-1})$ at $x\in\partial I$ belongs to $E(I)$ if the first $e_i$ not tangent to $\partial I$ points towards $I$. Define $\CA(I)$ to be the von Neumann algebra on $\CH_\CA$ generated by $\End_G(\bigotimes_{s\in P}W_s)$ where $P$ runs over all finite subsets of $E(I)$. 

\begin{prop}
The pair $(\CA,\CH_\CA)$, together with the obvious action of $\Diff(S^{n-1})$ on $\CA$, defines a finite irreducible topological $n$-net. 
\end{prop}
\begin{proof}
The axioms of locality, additivity, covariance and split property are all clear. To see the vacuum property we note that $\CH = \CH_\uparrow \otimes \CH_\uparrow^*$ where $\CH_\uparrow = \bigotimes_{s\in E(S^{n-1}_\uparrow)} W_s$ hence $\CH = L^2(\CL(\CH_\uparrow))$. Moreover, $\CL(\CH_\uparrow)$ is generated by $\End_\C(\bigotimes_{s\in P}W_s)$ for finite $P\subset E(S^{n-1}_\uparrow)$ hence $\CL_G(\CH_\uparrow)=\CA(S^{n-1}_\uparrow)$. Passing to the $G$-invariants yields $\CH_\CA = L^2(\CA(S^{n-1}_\uparrow))$. 
For the duality condition, note that for any disjoint $I,J\in\Disk^{n-1}$, the $\CO(\CH_\CA,I\cup J)$-$\CO(\CH_\CA,I'\cap J')^\op$-bimodule $\CH_A$ is semisimple hence dualizable.
It remains to show that $\CO(\CH_\CA) = \CL(\CH_\CA)$. Indeed, the operators from $\CO(\CH_\CA)$ supported on the point $(0,\dots,0,1)\in S^{n-1}$ intertwine $\CH_\uparrow$ and $\CH_\uparrow^*$. Together with these operators, $\CL_G(\CH_\uparrow)$ and $\CL_G(\CH_\uparrow^*)$ generate $\CL(\CH^G)$.
\end{proof}

\begin{rem}
The topological $n$-net $\CA$ constructed above (except the trivial case) differs from a conformal net on several aspects.
First, from the proof above we notice that all of $\CA(I)$ are type I von Neumann algebras on an inseparable Hilbert space while the local observable algebras of a conformal net are usually type III von Neumann algebras on a separable Hilbert space.
Second, $\CA$ does not satisfy the strong additivity axiom. %there are operators supported on any point of the $S^{n-1}$ thus 
Thirdly and most importantly, the whole diffeomorphism group $\Diff(S^{n-1})$ acts genuinely on $\CH_\CA$, however, the action is not continuous in any reasonable sense because every vector of $\CH_\CA$ lies over countably many points of $S^{n-1}$. In particular, 
%the conformal transformation group $\Conf(S^{n-1})$ 
$\Diff(S^{n-1})$ acts strongly continuously only on the one-dimensional subspace spanned by the vacuum vector $\bigotimes_{s\in E^+}(\mu\otimes\mu^*)$. 
%, in contrast to the fact that $\Conf(S^1)$ acts strongly continuously on the whole vacuum sector of a conformal net.
%the $L_0$ operator is only well-defined on a finite-dimensional subspace of $\CH_A$, in contrast to the fact that $L_0$ is densely defined for a conformal net. 
This suggests that the topological $n$-net $\CA$ describes the local quantum symmetry of a gapped phase.
\end{rem}

The above remark motivates the following definition. Typical examples include that a conformal net in the sense of \cite{BDH15} is a conformal 2-net and the topological $n$-net $\CA$ constructed above is gapped.

\begin{defn} \label{defn:net-gap}
We say that a defect $n$-net $\CA$ of codimension $k$ is {\em gapless} if $\Diff(S^{n-1})_k$ acts strongly continuously on an infinite-dimensional subspace of $\CH_\CA$. Otherwise, we say that $\CA$ is {\em gapped}.
We say that a topological $n$-net $\CA$ is a {\em conformal $n$-net} if $\Diff(S^{n-1})$ acts strongly continuously on $\CH_\CA$.
\end{defn}

\begin{rem}
If an $n$D CFT comes from a conformal $n$-net, then it admits not only an action of the conformal transformation group $\Conf(S^n)$ but also an action of the diffeomorphism group $\Diff(S^{n-1})$. This prompts a positive answer to the long standing question whether there are infinite-dimensional symmetries for higher dimensional CFT's as in the two dimensional case. We will come back to this issue in a subsequent work.
\end{rem}

Now we focus on dimension $n=2$ so that $E$ is a double cover of $S^1$. We use $x^\pm\in E^\pm$ to denote the point lying over $x\in S^1$.

We define a defect 2-net $\CF$ of codimension one between the trivial topological 2-net $\underline{\C}$ and $\CA$ (i.e. a boundary 2-net of $\CA$). 
Let $A$ be a nonzero $*$-Frobenius algebra in $\Rep G$ (in particular, $A$ is a separable algebra carrying a $G$-action). Assign $U_{(1,0)^+}=\C$ and $U_{(1,0)^-}=A^*$; $U_{(-1,0)^+}=A$ and $U_{(-1,0)^-}=\C$; $U_{x^\pm}=\C$ if $x_1<0$; $U_{x^+}=V^*$ and $U_{x^-}=V$ if $x_1>0$. Let $\CK=\bigotimes_{s\in E}U_s$
and define $\CH_\CF$ to be the subspace of $A$-invariants in $\CK^G$
$$\CH_\CF = (\CK^G)^A = \{ w\in\CK^G \mid a w=w a, \, \forall a\in A \}$$
where $A$ acts from the left on $A$ and from the right on $A^*$.

Let $I\in\Disk^1_1$ be an arc. Define $\CF(I)$ to be the von Neumann algebra on $\CH_\CF$ generated for all finite subsets $P\subset E(I)$ by the algebras $\End_G(\bigotimes_{s\in P}U_s)^A$ or $\End_G(\bigotimes_{s\in P}U_s)$ or $\C$ depending on the types of $U_s$.
Then $(\CF,\CH_\CF)$ is a finite defect 2-net of codimension one. It is irreducible if and only if $A$ is a simple $*$-Frobenius algebra in $\Rep G$. %, as seen from the following proposition.

\begin{prop}
$\Sect(\CF)$ is equivalent to the $*$-1-category of $A$-$A$-bimodules in $\widehat{\Rep G}$.
\end{prop}
\begin{proof}
Let $W=A\otimes(V\otimes V^*)^{\otimes k}\otimes A^*$ where $k\ge1$. Since $W$ contains all the simple $A$-$A$-bimodules in $\Rep G$, the category of $A$-$A$-bimodules in $\Rep G$ is equivalent to the category of finite-dimensional left $\End_{A|A}(W)^G$-modules. On the other hand, $\Sect(\CF)$ is equivalent to the category of left modules over the von Neumann algebra $\CL_{A|A}(\CK)^G$ which is Morita equivalent to $\End_{A|A}(W)^G$.
\end{proof}

\begin{cor}
$\Sect(\CA)\simeq\widehat{\FZ_1(\Rep G)}$.
\end{cor}
\begin{proof}
By folding $\CA$ along the line $x_1=0$ we obtain a defect 2-net between $\underline{\C}$ and $\CA\boxtimes\CA$ which shares the same sectors as $\CA$.
let $A=\tau\End_\C(V\otimes V^*)$ where $\tau$ is an adjoint functor to the tensor product $\otimes:\Rep G\boxtimes\Rep G\to\Rep G$.
Note that this defect 2-net coincides with the one constructed above for the finite group $G\times G$, the $(G\times G)$-module $V\otimes V^*$ and the $*$-Frobenius algebra $A$ in $\Rep G\boxtimes\Rep G \simeq \Rep(G\times G)$.
Since the category of right $A$-modules in $\Rep G\boxtimes\Rep G$ is equivalent to $\Rep G$, the category of $A$-$A$-bimodules in $\Rep G\boxtimes\Rep G$ is equivalent to the Drinfeld center $\FZ_1(\Rep G)$. Our claim follows.
\end{proof}

\begin{rem}
According to Remark \ref{rem:net-monoidal}, $\Sect(\CA)$ carries a braided monoidal structure. A direct way to see the braiding is to consider fusion of sectors as in \cite[Section 3]{BDH17}.
%We will see in Section \ref{sec:netn} that the topological 2-net $\CA$ defines an object of a $*$-3-category $\widehat\Net^2$ and it will be clear that $\Sect(\CA) = \Omega^2(\widehat\Net^2,\CA)$. Therefore, $\Sect(\CA)$ carries a braided monoidal structure. A direct way to see the braiding is to consider fusion of sectors as in \cite[Section 3]{BDH17}.
\end{rem}

\subsection{Levin-Wen nets} \label{sec:lw-net}

In this subsection, we provide the construction of a topological 2-net from a unitary fusion 1-category $\CC$ and a defect 2-net from a $*$-Frobenius algebra $A$ in $\CC$. They describe the Levin-Wen model associated to $\CC$ \cite{LW05} and the gapped boundary associated to $A$ \cite{KK12}. 

\smallskip

Let $\CC$ be a unitary fusion 1-category. Normalize the unit map $u_X:\one\to X\otimes X^*$ and the counit map $v_X:X^*\otimes X\to\one$ for every object $X\in\CC$ in such a way that the composite isomorphism induced by $v_X$ and $v_X^*$
$$\Hom_\CC(\one, X\otimes Y) \to \Hom_\CC(X^*,Y) \to \Hom_\CC(\one,Y\otimes X)$$
is isometric for all $Y\in\CC$. (In fact, these normalized duality maps induce a canonical spherical structure on $\CC$.)

Let $V\in\CC$ be an object which contains all the simple objects of $\CC$ as direct summands. Fix a morphism of norm one $\mu:\one\to V$.

Fix an orientation of the circle $S^1$ and choose a base point of $S^1$ so that the points of $S^1$ are linearly ordered. Assign $U_x=V\otimes V^*$ for all $x\in S^1$. For every finite subset $P \subset S^1$, we have an object $\bigotimes_{x\in P}U_x = U_{p_1}\otimes\cdots\otimes U_{p_k}$ of $\CC$ where $p_1,\dots,p_k$ are the points of $P$ in linear order. For an inclusion $P\subset Q$, the morphisms $\mu\otimes\mu^*:\one\to U_x$ for $x\in Q\setminus P$ induce a unitary embedding $\bigotimes_{x\in P}U_x \hookrightarrow \bigotimes_{x\in Q}U_x$.
Define $\bigotimes_{x\in S^1}U_x$ to be the direct limit $\varinjlim_P \bigotimes_{x\in P}U_x$ in $\hat\CC$.
We obtain a Hilbert space
$$\CH_\CA = \Hom_{\hat\CC}(\one,\bigotimes\nolimits_{x\in S^1}U_x) = \varinjlim_P \Hom_\CC(\one,\bigotimes\nolimits_{x\in P}U_x)$$
which is independent of the base point of $S^1$ due to the normalized duality maps.

Let $I\subset S^1$ be an arc from $a$ to $b$. Assign $W_a=V^*$, $W_b=V$ and $W_x=V\otimes V^*$ for $x\in I\setminus\{a,b\}$. Define $\CA(I)$ to be the type I von Neumann algebra
$$\CA(I) = \Hom_{\hat\CC}(\bigotimes\nolimits_{x\in I}W_x,\bigotimes\nolimits_{x\in I}W_x)$$
where $\bigotimes_{x\in I}W_x$ is an object of $\hat\CC$ defined by a direct limit as above.
The action of $\CA(I)$ on $\bigotimes_{x\in S^1}U_x$ induces a left action on $\CH_\CA$.

Then $(\CA,\CH_\CA)$ is a finite irreducible topological 2-net. Indeed, in the special case $\CC=\Rep G$ we recover the topological 2-net defined in the previous subsection. In general, the proof is essentially the same.

Moreover, by choosing a nonzero $*$-Frobenius algebra $A$ in $\CC$, one defines similarly a finite defect 2-net $\CF$ of codimension one between the trivial topological 2-net $\underline{\C}$ and $\CA$ (i.e. a boundary 2-net of $\CA$). 
We conclude by similar arguments that $\Sect(\CF)$ is equivalent to the $*$-1-category of $A$-$A$-bimodules in $\hat\CC$ and that $\Sect(\CA) \simeq \widehat{\FZ_1(\CC)}$.
One can say more precisely about the latter equivalence: an object $Z\in\FZ_1(\CC)$ induces a sector $\Hom_{\hat\CC}(\one,Z\otimes\bigotimes_{x\in S^1}U_x)$ of $\CA$.

\smallskip

The construction can be generalized to higher unitary fusion categories. We will address the problem in a more general context of condensation theory in Section \ref{sec:cond-nnet}.

\section{Construction of categories of quantum liquids}

%Unless specified explicitly, quantum liquids and domains walls are anomaly-free.

\subsection{Categories of topological nets} \label{sec:netn}

Let $S^n_+$ denote the positive half of $S^n$ defined by $x_n\ge0$ and let $S^n_-$ denote the negative half of $S^n$ defined by $x_n\le0$ so that $S^{n-1} = S^n_+\cap S^n_-$.

We claim the following result:

\begin{thm} \label{thm:netn}
There is a symmetric monoidal $*$-$(n+1)$-category $\widehat\Net^n$ containing the following pieces of data:
\begin{itemize}
\item An object is a topological $n$-net $\CA$. 

\item The tensor product of two objects $\CA,\CB$ is $\CA\boxtimes\CB$ and the direct sum of $\CA,\CB$ is $\CA\oplus\CB$. The tensor unit is the trivial topological $n$-net $\underline{\C}$ and the zero object is $\underline{0}$.

\item A $k$-morphism $\CF:\CA\to\CB$ for $1\le k<n$ is a defect $n$-net of codimension $k$ such that
\begin{equation*}
\begin{split}
& \CF(I)=\CA(I),\quad \text{if $I\cap S^{n-k}_+=\emptyset$}, \\
& \CF(I)=\CB(I),\quad \text{if $I\cap S^{n-k}_-=\emptyset$},
\end{split}
\end{equation*}
and, moreover, the isomorphism $\eta_{h,I}$ associated to $\CF$ agrees with that to $\CA$ or $\CB$ if $I\cap S^{n-1-k}=\emptyset$.

\item
The identity $k$-morphism $\Id_\CA:\CA\to\CA$ is the restriction of $\CA$ (see Remark \ref{rem:res-k-k+1}). The direct sum of two $k$-morphisms $\CF,\CG:\CA\to\CB$ is determined by the formula
$$(\CF\oplus\CG)(I) = \CF(I)\oplus\CG(I), \quad \text{if $I\cap S^{n-1-k}\ne\emptyset$}.$$

\item An $n$-morphism $\CH:\CA\to\CB$ is a sector of the partial $n$-net $\CS_{\CB|\CA}:\Disk^{n-1}_{n-1}\to\VN$ defined by
\begin{equation*}
\begin{split}
& \CS_{\CB|\CA}(I)=\CA(I),\quad \text{if $I\cap S^0_+=\emptyset$}, \\
& \CS_{\CB|\CA}(I)=\CB(I),\quad \text{if $I\cap S^0_-=\emptyset$}.
\end{split}
\end{equation*}

\item The identity $n$-morphism $\Id_\CA:\CA\to\CA$ is the vacuum sector $\CH_\CA$.
The composition of two $n$-morphisms $\CH:\CA\to\CB$ and $\CK:\CB\to\CC$ is the Connes fusion $\CK\boxtimes_{\CB(S^{n-1}_\uparrow)}\CH$ (recall that $\CB(S^{n-1}_\downarrow)\cong\CB(S^{n-1}_\uparrow)^\op$ canonically).

\item An $(n+1)$-morphism is a homomorphism of sectors. The $*$-involution sends $f:\CH\to\CK$ to the adjoint $f^*:\CK\to\CH$.
\end{itemize}
Moreover, the collection of finite topological $n$-nets, finite defect $n$-nets and semisimple sectors form a symmetric monoidal $*$-subcategory $\Net^n$ of $\widehat\Net^n$.
\end{thm}

We refer the reader to \cite{BDH15,BDH17,BDH19a,BDH18} for a detailed construction of a symmetric monoidal $*$-3-category of conformal nets. A bunch of techniques have been developed there and can be generalized to higher dimensions. In particular, the fusion operation of defects defined in \cite{BDH19a}, after generalized to higher dimensions, implements a crucial ingredient of the higher category $\widehat\Net^n$: the composition of $k$-morphisms for $1\le k<n$.
According to Example \ref{exam:conf-net} and Example \ref{exam:conf-defect}, the $*$-3-category of conformal nets is a full subcategory of $\widehat\Net^2$.

We will sketch a construction of $\widehat\Net^n$ and $\Net^n$ in Section \ref{sec:fus-def}.

\begin{exam}
Since $\Disk^{-1}$ is an empty 1-category, a 0-net is an empty functor of which a sector is simply a Hilbert space. Therefore, we may define $\widehat\Net^0=\widehat\Hilb$ and $\Net^0=\Hilb$.
\end{exam}

\begin{exam} \label{exam:net1}
According to Example \ref{exam:1tn}, giving a topological 1-net is equivalent to giving a von Neumann algebra. Therefore, $\widehat\Net^1$ can be identified with the symmetric monoidal $*$-2-category of von Neumann algebras and bimodules; $\Net^1$ can be identified with the symmetric monoidal $*$-2-category of finite direct sums of type I factors and semisimple bimodules.
\end{exam}

\begin{rem} \label{rem:net-monoidal}
Note that $\Sect(\CA) = \Omega^n(\widehat\Net^n,\CA)$ for any topological $n$-net $\CA$. Therefore, $\Sect(\CA)$ carries an $E_n$-monoidal structure. Similarly, for any $k$-morphism $\CF$ of $\widehat\Net^n$ where $1\le k<n$, $\Sect(\CF)$ carries an $E_{n-k}$-monoidal structure.
\end{rem}

According to Example \ref{exam:tn2dn}, an object of $\widehat\Net^n$ induces a 1-morphism of $\widehat\Net^{n+1}$ between the tensor unit $\underline{\C}$. Conversely, by the additivity axiom, all of the 1-morphisms between $\underline{\C}$ arise in this way.
Therefore, we have identifications
$$\widehat\Net^n = \Omega\widehat\Net^{n+1}, \quad\quad \Net^n = \Omega\Net^{n+1}.$$

\begin{rem}
Bartels, Douglas and Henriques answered positively in their works \cite{BDH15,BDH17,BDH19a,BDH18} the following question proposed by Stolz and Teichner: Does there exist an interesting 3-category that deloops the 2-category of von Neumann algebras? They showed that the $*$-3-category of conformal nets is a delooping of the $*$-2-category of von Neumann algebras. Theorem \ref{thm:netn} improves their result to that $\widehat\Net^n$ is an $(n-1)$-fold delooping of the $*$-2-category of von Neumann algebras (as well as an $n$-fold delooping of the $*$-1-category of Hilbert spaces).
\end{rem}

\subsection{Construction of $\TO^n$} \label{sec:construct-QL}

Let us assume the higher categories $\Net^n$ are well-defined. We construct explicitly higher categories of quantum liquids $\TO^n$ and give in particular a precise mathematical definition of a quantum liquid.

\smallskip

We define first for each $n$ a symmetric monoidal $*$-$(n+1)$-subcategory $\LQS^n \subset \Net^n$ to describe local quantum symmetries as follows.
First, let 
$$\LQS^0 = \Net^0 = \Hilb.$$ 
Then by induction on $n$, define $\LQS^n \subset \Net^n$ to be the maximal subcategory extending the embedding $B\LQS^{n-1} \hookrightarrow \Sigma_*\LQS^{n-1}$. In another word, $\LQS^n$
is obtained from the subcategory $B\LQS^{n-1} \subset \Net^n$ by appending all $*$-condensates.
By the construction, there is a forgetful functor
$$\sk: \LQS^n \to (n+1)\Hilb.$$

We define then
$$\TO^n = \begin{cases}
  \underline{\C}/\LQS^n, & \text{for even $n$}, \\
  \LQS^n/\underline{\C}, & \text{for odd $n$}. \\
\end{cases}
$$
It comes equipped with forgetful functors
$$\sk: \TO^n \to \TOsk^n,$$
$$\lqs: \TO^n \to \LQS^n.$$
Recall that $\TOsk^n = \Sigma_*^n(\C/\Hilb)$. According to \cite[Proposition 4.12]{KZ20b}, we may identify
$$\TOsk^n = \begin{cases}
  \bullet/(n+1)\Hilb, & \text{for even $n$}, \\
  (n+1)\Hilb/\bullet, & \text{for odd $n$}. \\
\end{cases}
$$

An object of $\TO^n$ is referred to as an $n$D {\em quantum liquid}. A $k$-morphism $\CF:\CA\to\CB$ of $\TO^n$ where $1\le k< n$ is referred to as a {\em domain wall} between $\CA$ and $\CB$ or a {\em defect} of codimension $k$. An $n$-morphism of $\TO^n$ is referred to as an {\em instanton}. 
The tensor unit $\Id_{\underline{\C}}$ of $\TO^n$ is denoted by $\one^n$ and referred to as the $n$D {\em trivial quantum liquid}.
We say that a domain wall is {\em trivial} if it is an identity morphism; {\em invertible} if it is an invertible morphism of $\TO^n$.

For an object or a morphism $\CX$ of $\TO^n$, we refer to the image $\sk(\CX)$ in $\TOsk^n$ as the {\em topological skeleton} of $\CX$, and refer to the image $\lqs(\CX)$ in $\LQS^n$ as the {\em local quantum symmetry} of $\CX$.

We say that a quantum liquid or defect $\CX$ is {\em gapped} (resp. {\em gapless}) if the defect net $\lqs(\CX)$ is gapped (resp. gapless).

\begin{exam}
We have $\TO^0 = \TOsk^0 = \C/\Hilb$.
\end{exam}

\begin{exam} \label{exam:to1}
By Example \ref{exam:net1}, we have $\LQS^1 = \Net^1$ which can be identified with the symmetric monoidal $*$-2-category of finite direct sums of type I factors and semisimple bimodules. Moreover, the forgetful functor $\LQS^1\to2\Hilb$ maps an object $\CA$ to the unitary 1-category $\LMod_\CA^s$ of semisimple left $\CA$-modules.
Therefore, $\TO^1$ has the following structures:
\begin{itemize}
\item An object $(\CA,\CH)$ consists of a finite direct sum of type I factors $\CA$ and a semisimple right $\CA$-module $\CH$. 
\item A 1-morphism $(\CF,f):(\CA,\CH)\to(\CB,\CK)$ consists of a semisimple $\CB$-$\CA$-bimodule $\CF$ and a right $\CA$-module map $f:\CH\to\CK\boxtimes_\CB\CF$.
\item A 2-morphism $\xi:(\CF,f)\to(\CG,g)$ is a bimodule map $\xi:\CF\to\CG$ such that $(\CK\boxtimes_\CB\xi)\circ f=g$.
\end{itemize}
Note that $\lqs(\CA,\CH) = \CA$ and that $\sk(\CA,\CH)$ is the functor $\CH\boxtimes_\CA-: \LMod_\CA^s\to\Hilb$. A 1D quantum liquid $(\CA,\CH)$ is gapped if and only if $\CA$ is finite-dimensional.
\end{exam}

We claim the following result which is requisite by physics (\cite[Theorem 5.13]{KZ20b}):

\begin{thm} \label{thm:lqs-cc}
$\LQS^n$ is $*$-condensation-complete. Consequently, the following forgetful functors
$$\sk:\LQS^n \to (n+1)\Hilb,$$
$$\sk:\TO^n \to \TOsk^n$$
are symmetric monoidal equivalences.
\end{thm}

The theorem is true for $n=1$ as clear from Examples \ref{exam:to1}.
However, the theorem becomes nontrivial for $n=2$. The Levin-Wen 2-nets constructed in Section \ref{sec:lw-net} show that the forgetful functor $\sk:\LQS^2 \to 3\Hilb$ is essentially surjective. But this is not sufficient to establish the theorem for $n=2$. %In general there are more than one preimage of an object of $3\Hilb$; in particular a preimage might be a conformal net. 
One needs to show that there are sufficiently many 1-morphisms in $\LQS^2$.
The problem is even more complicated for $n=3$: 3D Chern-Simons theory and 3D CFTs are expected to be involved in $\LQS^3$.

We will sketch a proof of Theorem \ref{thm:lqs-cc} in Section \ref{sec:condense-net}.

\begin{rem}
Let $\CA\in\LQS^n$. Then $\sk(\CA)\in(n+1)\Hilb$ is a unitary $n$-category. Let $\CB$ denote the unitary braided multi-fusion $(n-1)$-category $\Omega\Fun(\sk(\CA),\sk(\CA))$. Since $\sk: \LQS^n \to (n+1)\Hilb$ is an equivalence, $\CB \simeq \Omega\Hom_{\LQS^n}(\CA,\CA) = \Hom_{\LQS^n}(\Id_\CA,\Id_\CA)$. In particular, an object of $\CB$ corresponds to a 2-morphism $\Id_\CA\to\Id_\CA$ of $\LQS^n$ which is a defect $n$-net of codimension two for large $n$. %More generally, a 1-morphism of $\CB$ is a 3-morphism of $\LQS^n$ and so on.

Let $\CX:\underline{\C}\to\CA$ be an object of $\TO^n$ where $n$ is even (the odd case is similar). Regard $\sk(\CX)$ as an object of $\sk(\CA)$ and let $\CC$ denote the unitary multi-fusion $(n-1)$-category $\Omega(\sk(\CA),\sk(\CX)) = \Hom_{\sk(\CA)}(\sk(\CX),\sk(\CX))$. Since $\sk: \LQS^n \to (n+1)\Hilb$ is an equivalence, $\CC \simeq \Hom_{\LQS^n}(\CX,\CX)$. In particular, an object of $\CC$ corresponds to a 2-morphism $\CX\to\CX$ of $\LQS^n$ which is a defect $n$-net of codimension two for large $n$. %More generally, a 1-morphism of $\CC$ is a 3-morphism of $\LQS^n$ and so on.

Consider the generic case where $\sk(\CA)$ is indecomposable and $\sk(\CX)$ is nonzero. We have $\Sigma_*\CC = \sk(\CA)$ by \cite[Corollary 3.13]{KZ20b} so that $\CB=\FZ_1(\CC)$ by \cite[Theorem 3.41]{KZ20b}. In particular, one recovers $\sk(\CX)$ as the object $\CC$ of $\Sigma_*\CC$ so that we can say that $\CC$ is the topological skeleton of the quantum liquid $\CX$.
Moreover, $\CC$ is enriched in $\Hom_{\LQS^n}(\Id_\CA,\Id_\CA)$ in the sense of \cite[Definition 3.25]{KZ20b} via the equivalence $\CB \simeq \Hom_{\LQS^n}(\Id_\CA,\Id_\CA)$. Note that the topological $n$-net $\CA=\lqs(\CX)$ encodes the local observable algebras living on the space-time and, moreover, it determines both $\Hom_{\LQS^n}(\Id_\CA,\Id_\CA)$ and the braided monoidal equivalence $\CB \simeq \Hom_{\LQS^n}(\Id_\CA,\Id_\CA)$. Therefore, $\lqs(\CX)$ is exactly what we have expected in \cite[Section 5.2]{KZ20b} for the local quantum symmetry of the quantum liquid $\CX$.

However, the content of the quantum liquid $\CX$ is slightly more than we have expected in \cite[Section 5.2]{KZ20b}. Besides the topological skeleton $\sk(\CX)$ and the local quantum symmetry $\lqs(\CX)$, there are also local observable algebras of the boundary $n$-net $\CX$.
%Besides the topological skeleton $\sk(\CX)$ and the local quantum symmetry $\lqs(\CX)$, there are local observable algebras living on the space encoded in the boundary $n$-net $\CX$ of $\lqs(\CX)$. 
\end{rem}

\begin{rem}
Recall that a topological $n$-net $\CA\in\LQS^n$ can be viewed as a 1-morphism $\CA:\underline{\C}\to\underline{\C}$ of $\LQS^{n+1}$, i.e. an $n+1$D quantum liquid whose local quantum symmetry is trivial. This provides an explanation of the topological Wick rotation introduced in \cite{KZ20a}. Namely, the local quantum symmetry of an $n$D quantum liquid $\CX:\underline{\C}\to\CA$ shares the same mathematical content as an $n+1$D quantum liquid with trivial local quantum symmetry.
%The topological Wick rotation is a trivial consequence of our definition of a quantum liquid: If $\CX$ is an $n$D quantum liquid then $\lqs(\CX)$ defines an $(n+1)$D quantum liquid with trivial local quantum symmetry of which $\CX$ is a boundary.
\end{rem}

Before concluding this subsection, we should point out that the $*$-involution of $\TO^n$ is not the one induced by the time-reversal operator. Let us endow $\LQS^n$ with a new involution $*: \LQS^n \to (\LQS^n)^{\op n}$ which still fixes all the objects and all the $(n-1)$- and lower morphisms but sends an $n$-morphism $\CH$ to the complex conjugate $\bar\CH$ and an $(n+1)$-morphism $f:\CH\to\CK$ to $\bar f:\bar\CH\to\bar\CK$. Then endow $\TO^n$ with the induced $*$-involution.

\subsection{Transparent domain walls} \label{sec:transparent-wall}
%\subsection{Detecting local quantum symmetries}

As we have seen from the previous subsection, $\TO^n$ is equivalent to $\TOsk^n$.
That is, local quantum symmetries can not be recovered from the equivalence type of $\TO^n$ at all.

To detect local quantum symmetries, one needs to distinguish a class of transparent domain walls from invertible ones in certain categorical structures.
Roughly speaking, a domain wall $\CF$ between two quantum liquids or defects $\CA$ and $\CB$ is transparent if $\CA$ and $\CB$ can be identified in such a way that $\CF$ is a trivial domain wall.

\smallskip

We say that two defect $n$-nets $\CA$ and $\CB$ of codimension $k$ are {\em isomorphic} if there exists a natural isomorphism $\xi:\CA\to\CB$ intertwining the action of $\Diff(S^{n-1})_k$ such that $\xi_I$ is the identity for $I\cap S^{n-1-k}=\emptyset$.

Then we say that a $k$-morphism $\CF:\CA\to\CB$ of $\Net^n$ where $1\le k\le n$ is {\em completely transparent} if the defect $n$-nets $\CA$ and $\CB$ can be identified by an isomorphism in such a way that the defect $n$-net $\CF$ is isomorphic to the identity $k$-morphism.
An $(n+1)$-morphism of $\Net^n$ is {\em completely transparent} if it is an isometry of sectors.

Let $\CT$ be the minimal class of morphisms of $\Net^n$ that contains all the completely transparent morphisms and is closed under composition, tensor product and adjunction. The members of $\CT$ are called {\em transparent}.

\begin{exam}
A 1-morphism $\CF:\CA\to\CB$ of $\Net^1$ is (completely) transparent if and only if they are isomorphic as von Neumann algebras. 
A 2-morphism of $\Net^1$ is (completely) transparent if and only if it is an isometry of bimodules.
\end{exam}

We say that a morphism of $\TO^n$ is {\em transparent} if it is defined in terms of transparent morphisms of $\LQS^n$. We say that two objects or morphisms of $\TO^n$ are {\em unitarily equivalent} if there is a transparent morphism between them.

\begin{exam}
In the notation of Example \ref{exam:to1}, two objects $(\CA,\CH)$ and $(\CB,\CK)$ of $\TO^1$ are unitarily equivalent if and only if there exists an isomorphism $\CA\cong\CB$ and an isometry of right $\CA$-modules $\CH\cong\CK$. 
Two 1-morphisms $(\CF,f)$ and $(\CG,g)$ are unitarily equivalent if and only if there is an isometry of bimodules $\xi:\CF\to\CG$ such that $(\CK\boxtimes_\CB\xi)\circ f=g$.
\end{exam}

Define $\widetilde\LQS^n_k \subset \LQS^n$ to be the symmetric monoidal subcategory obtained by discarding all the nontransparent (k+1)- and higher morphisms. Similarly, define $\widetilde\TO^n_k \subset \TO^n$ be the symmetric monoidal subcategory obtained by discarding all the nontransparent (k+1)- and higher morphisms.
We obtain two towers of symmetric monoidal $(n+1)$-categories
$$\widetilde\LQS^n_0 \subset \widetilde\LQS^n_1 \subset \cdots \subset \widetilde\LQS^n_{n+1} = \LQS^n,$$
$$\widetilde\TO^n_0 \subset \widetilde\TO^n_1 \subset \cdots \subset \widetilde\TO^n_{n+1} = \TO^n$$
to encode the information of local quantum symmetries. Indeed, two $n$D quantum liquids are unitarily equivalent if and only if they are equivalent in $\widetilde\TO^n_0$. Two defects of codimension $k$ are unitarily equivalent if and only if they are equivalent in $\widetilde\TO^n_k$.

\begin{rem} [Holographic principle] \label{rem:holo}
In \cite{KWZ15,KWZ17}, a new kind of morphisms between potentially anomalous quantum liquids were introduced. In the present context, they are nothing but the 1-morphisms of the coslice $(n+1)$-category $\one^n/\widetilde\TO^n_1$. An object of $\one^n/\widetilde\TO^n_1$ is referred to as an $(n-1)$D {\em potentially anomalous quantum liquid} which by definition is a 1-morphism $\CX:\one^n\to\CA$ of $\widetilde\TO^n_1$. Note that $\CA$ is an $n$D anomaly-free quantum liquid and $\CX$ is a boundary of $\CA$. Also note that the holographic principle holds trivially: the bulk $\CA$ is uniquely determined by the boundary $\CX$. However, one can say more concretely about this principle. The reasoning in \cite{KWZ15,KWZ17} shows that $\CA$ is the center of $\CX$ in the sense that the trivial domain wall of $\CA$, viewed as an $(n-1)$D potentially anomalous quantum liquid by folding $\CA$, satisfies the universal property of the center of $\CX$. 
%We are interested particularly in the Wick-rotated version: $$\underline{\C}/\widetilde\LQS^n_1.$$ By definition, an object of $\underline{\C}/\widetilde\LQS^n_1$ is a 1-morphism $\CX:\underline{\C}\to\CA$ of $\widetilde\LQS^n_1$. Namely, $\CA$ is a topological $n$-net and $\CX$ is a boundary $n$-net of $\CA$. The reasoning in \cite{KWZ15,KWZ17} shows that $\CA$ is the center of $\CX$ in the sense that $\Id_\CA$ satisfies the universal property of the center of $\CX$. In particular, the bulk $\CA$ is uniquely determined by the boundary $\CX$ at the categorical level (forgetting the intrinsic data of $\CA$ and $\CX$).
\end{rem}

%\subsection{Topological phase transition}

\section{Construction of categories of topological nets} \label{sec:fus-def}

The purpose of this section is to sketch a construction of $\widehat\Net^n$ and $\Net^n$.

\subsection{Nets on manifolds}

Unless specified explicitly, manifolds are smooth, oriented, compact without boundary and possibly disconnected.

Let $M$ be an $n$-manifold. A {\em disk region} of $M$ is a closed region that is diffeomorphic to the $n$-disk. Let $\Disk(M)$ denote the 1-category of disk regions of $M$ whose morphisms are inclusions of disk regions. 

A {\em partial net} on $M$ is a functor $\CA:\CD\to\VN$ where $\CD$ is a full subcategory of $\Disk(M)$. 
A {\em sector} of a partial net $\CA:\CD\to\VN$ is defined literally as Definition \ref{defn:sector}.

\begin{defn}
A {\em stratified net} on $M$ consists of the following data:
\begin{itemize}
\item a partial net $\CA:\CD\to\VN$ on $M$;
\item a stratification $M=M_n\supset M_{n-1}\supset\cdots\supset M_0$ where $M_i$ is a closed submanifold of dimension $i$ around which $M_{i+1}$ is equipped with a local parametrization.
\end{itemize}
These data are subject to the following condition:
\begin{itemize}
%\item The interiors of $I\in\CD$ form a basis for the topology of $M$.
%For every disk region $J\subset M$, there exist a family of disk regions $I_\alpha\in\CD$ such that the interior of $J$ is the union of the interiors of $I_\alpha$ and that $\partial J$ is the union of the interiors of $\partial I_\alpha\cap\partial J$ in $\partial J$. 
\item For every $I\in\CD$, there exists a diffeomorphism $h:S^n_\uparrow\to I$ such that $h(S^{n-i}_\uparrow)=I\cap M_{n-i}$ and $h$ transforms the local parametrization around $S^{n-i}_\uparrow$ into that around $I\cap M_{n-i}$ for $1\le i\le d(I)$ where $d(I)$, called the {\em depth} of $I$, is the maximal integer such that $I\cap M_{n-d(I)}\ne\emptyset$.
%For every $I\in\CD$, there exists a diffeomorphism $h:S^n_\uparrow\to I$ such that $h(S^i_\uparrow)=I\cap M_i$ and $h$ transforms the local parametrization around $S^i_\uparrow$ into that around $I\cap M_i$ for $n-d(I)\le i\le n-2$ where $d(I)$, called the {\em depth} of $I$, is the maximal integer such that $I\cap M_{n-d(I)}\ne\emptyset$.
%For every $I\in\CD$, if $I$ intersects $M_i$ then there exists a diffeomorphism $h:S^n_\uparrow\to I$ such that $h(S^j_\uparrow)=I\cap M_j$ and $h$ transforms the local parametrization around $S^j_\uparrow$ into that around $I\cap M_j$ for $i\le j\le n-2$. The maximal integer $k$ such that $I\cap M_{n-k}\ne\emptyset$ is called the {\em depth} of $I$.
\end{itemize}
We use the notation $\CA$ to denote a stratified net.
\end{defn}

\begin{exam}
A defect $n$-net of codimension $k$ defines a stratified net on $S^{n-1}$ with respect to the stratification $S^{n-1} \supset\cdots\supset S^{n-1-k} \supset \emptyset \supset\cdots\supset \emptyset$ and the standard local parametrizations.
\end{exam}

\subsection{Sewing operations}

Let $\CA_\pm:\CD_\pm\to\VN$ be a pair of stratified nets on $(n-1)$-manifolds $M_\pm$ and let $\CB$ be a defect $n$-net of codimension $k$. We consider an operation of sewing $\CA_\pm$ along $\CB$ as follows.

Define subsets of $S^{n-1}$:
$$S_\pm = \{ x\in S^{n-1} \mid \pm x_{n-1-k}\ge0 \}, \quad S_0=S_+\cap S_-,$$
$$S_{\pm\epsilon} = \{ x\in S^{n-1} \mid \pm x_{n-1-k}\ge-\epsilon \}, \quad S_\epsilon = S_{+\epsilon}\cap S_{-\epsilon},$$
where $\epsilon>0$ is a small real number fixed once and for all. 

Suppose given a pair of disk regions $I_\pm\in\CD_\pm$ of depth $k$, a pair of diffeomorphisms $\phi_\pm: S_{\pm\epsilon} \to I_\pm$ that preserve the orientation, the stratification and the local parametrizations, and a pair of natural isomorphisms $\CA_\pm(\phi_\pm(K_\pm)) \cong \CB(K_\pm)$ for $S_{\pm\epsilon}\supset K_\pm\in\Disk^{n-1}_k$ (in particular, $\phi_\pm(K_\pm)\in\CD_\pm$). 

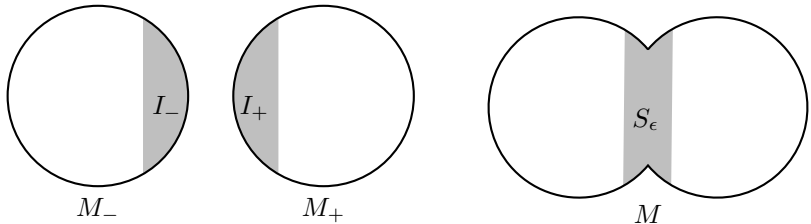
\medskip
\begin{center}
\begin{tikzpicture}[scale=1.2]
\fill[color=lightgray] (1.5,0.866) arc (60:-60:1) node[above=2.5em,right=0em][color=black]{$I_-$};
\draw[thick] (1,0) circle (1) node[below=3.5em]{$M_-$};
\fill[color=lightgray] (3,0.866) arc (120:240:1) node[above=2.5em,left=0em][color=black]{$I_+$};
\draw[thick] (3.5,0) circle (1) node[below=3.5em]{$M_+$};
\end{tikzpicture}
\quad\quad
\begin{tikzpicture}[scale=1.2]
\fill[color=lightgray] (-0.257,0.866) arc (60:40:1) arc (140:120:1) -- (0.257,-0.866) arc (240:220:1) arc (-40:-60:1) node[above=2.5em,right=0em][color=black]{$S_\epsilon$};
\draw[thick] (0,0.643) arc (40:320:1) arc (-140:140:1) node[below=5.5em]{$M$};
\end{tikzpicture}
\end{center}

Gluing $M_\pm\setminus\phi_\pm(S_\pm)$ and $S_\epsilon$ via the diffeomorphisms $\phi_\pm$, we obtain an $(n-1)$-manifold 
$$M = (M_-\setminus\phi_-(S_-)) \cup S_\epsilon \cup (M_+\setminus\phi_+(S_+))$$ 
which inherits a stratification from $M_\pm$. Let $\CD\subset\Disk(M)$ be the full subcategory consisting of the disk regions $M_\pm\setminus\phi_\pm(S_\pm) \supset J_\pm\in\CD_\pm$ and $S_\epsilon\supset J\in\Disk^{n-1}_k$. Assembling the algebras $\CA_\pm(J_\pm)$ and $\CB(J)$ then yields a stratified net on $M$
$$\CA:\CD\to\VN.$$

\smallskip

We consider next an operation of sewing a pair of sectors $\CH_\pm$ of $\CA_\pm$. Let 
$$\CH = \CH_-\boxtimes_{\CB(S_+)}\CH_+.$$ 
Then $\CH$ carries the structure of a sector of $\CA$ in the following way. On the one hand, $\CA_\pm(J_\pm)$ acts on $\CH$ via $\CH_\pm$ for $M_\pm\setminus\phi_\pm(S_\pm) \supset J_\pm\in\CD_\pm$. On the other hand, we have $$\CH \cong \CH_- \boxtimes_{\CB(S_{+\epsilon})} L^2(\CB(S_{+\epsilon})) \boxtimes_{\CB(S_+)} L^2(\CB(S_{+\epsilon})) \boxtimes_{\CB(S_{+\epsilon})} \CH_+$$ and $L^2(\CB(S_{+\epsilon})) \boxtimes_{\CB(S_+)} L^2(\CB(S_{+\epsilon})) \cong \CH_\CB \boxtimes_{\CB(S_+)} \CH_\CB \cong \CH_\CB$. Then $\CB(J)$ acts on $\CH$ via $\CH_\CB$ for $S_\epsilon\supset J\in\Disk^{n-1}_k$.
%Then, by identifying the left $\CB(S_{\pm\epsilon})$-module $\CH_\pm$ as a direct summand of a direct sum of $\CH_\CB$ (see Lemma \ref{lem:lmod-l2}), we obtain an action of $\CB(J)$ on $\CH$ for $S_\epsilon\supset J\in\Disk^{n-1}_k$.
%(Here, we need to choose an isomorphism of left $\CB(S_{+\epsilon})$-modules $L^2(\CB(S_{+\epsilon})) \cong \CH_\CB$ but the action of $\CB(J)$ does not depend on this choice.)

\begin{exam}[Fusion of sectors]
The composition of two $n$-morphisms $\CH_+:\CP\to\CB$ and $\CH_-:\CB\to\CQ$ of $\widehat\Net^n$ is precisely defined by sewing the sectors $\CH_\pm$ along $\CB$ where $\CA_+=\CS_{\CB|\CP}$, $\CA_-=\CS_{\CQ|\CB}$, $M_\pm=S^{n-1}$ and $\phi_\pm=\Id_{S_{\pm\epsilon}}$. Indeed, $\CS_{\CQ|\CP}$ is a completion of $\CA$ so that $\CH=\CH_-\boxtimes_{\CB(S_\uparrow)}\CH_+$ is a sector of $\CS_{\CQ|\CP}$.
\end{exam}

\void{
The above sewing operation generalizes readily to the case where $I_\pm$ are disjoint disk regions of a single manifold. The sector sewing is implemented by the completely additive functor (see Proposition \ref{prop:nfun-bimod}) $$L^2(\CB(S_{-\epsilon}) \bar\otimes \CB(S_{+\epsilon})) \mapsto L^2(\CB(S_{-\epsilon})) \boxtimes_{\CB(S_+)} L^2(\CB(S_{+\epsilon})).$$
}

The above sewing operation can be generalized to sew $\CA_\pm$ along $\CB$ over a generic closed region $R\subset S^{n-1}$ with smooth boundary instead of a disk region. The local observable algebras around the seam are defined to be those of $\CB$ around $\partial R$.

\begin{prop} \label{prop:sec-ss-dual}
Let $\CS_{\CB|\CA}$ be the partial $n$-net from the definition of an $n$-morphism $\CA\to\CB$ of $\Net^n$. Then $\Sect(\CS_{\CB|\CA})$ is equivalent to a finite direct sum of $\widehat\Hilb$. A sector of $\CS_{\CB|\CA}$ is semisimple if and only if it is a dualizable $\CB(S^{n-1}_\uparrow)$-$\CA(S^{n-1}_\uparrow)$-bimodule.
\end{prop}
\begin{proof}[Sketch of proof]
The proof follows the lines of \cite[Theorem 3.14]{BDH15} (see also \cite{KLM01} and \cite[Theorem 1.10]{BDH19b}).
Replacing $\CA$ and $\CB$ by $\CA\oplus\CB$ if necessary, we may assume without loss of generality that $\CA=\CB$. Then $\CS_{\CB|\CA}=\CA$.
Sewing $\CA$ and $\CA^\op$ over the cylinder $\{x\in S^{n-1} \mid |x_0|\le1/2]\}$, we obtain a stratified net $\CS$ as well as a sector $\CH$, where $\CS$ is a disjoint union of two copies of $\CA$. By the split property and the duality condition, $\CH$ is a dualizable $\CA(S^{n-1}_\uparrow)\bar\otimes\CA(S^{n-1}_\downarrow)$-$\CA(S^{n-1}_\uparrow)\bar\otimes\CA(S^{n-1}_\downarrow)$-bimodule. Note that the duality maps can be chosen to be homomorphisms of sectors of $\CS$. Therefore, $\CH$ has to be a semisimple sector of $\CS$. That is, $\CH$ has the form $\CH_1\otimes\CK_1\oplus\cdots\oplus\CH_m\otimes\CK_m$ where $\CH_i$ and $\CK_i$ are irreducible sectors of $\CA$. Moreover, $\CH_i$ and $\CK_i$ are dualizable $\CA(S^{n-1}_\uparrow)$-$\CA(S^{n-1}_\uparrow)$-bimodules.
Applying an argument as the proof of \cite[Theorem 3.14]{BDH15}, we see that every sector of $\CA$ is a direct sum of possibly infinite copies of the $\CH_i$'s.
\end{proof}

\begin{cor} \label{cor:net-sec-dual}
The composition of $n$-morphisms of $\Net^n$ is well-defined. Moreover, all the $n$-morphisms of $\Net^n$ are dualizable.
\end{cor}

\subsection{Fusion of defects}

Let $\CA_+:\CP\to\CB$ and $\CA_-:\CB\to\CQ$ be two $(k+1)$-morphisms of $\widehat\Net^n$ where $0\le k\le n-2$. We construct the composition $\CA_-\circ\CA_+:\CP\to\CQ$ by using the sewing operation defined in the previous subsection.

\smallskip

Let $M_\pm=S^{n-1}$ and $\phi_\pm: S_{\pm\epsilon} \to \overline{S^{n-1}\setminus S_{\mp\epsilon}}$ be the dilation map. %along the $x_{n-1-k}$-axis. 
Sewing $\CA_\pm$ along $\CB$ via $\phi_\pm$ then yields a stratified net $\CA:\CD\to\VN$ on an $(n-1)$-manifold $M$ which is diffeomorphic to $S^{n-1}$. Moreover, we have a sector of $\CA$ $$\CH = \CH_{\CA_-}\boxtimes_{\CB(S_+)}\CH_{\CA_+}.$$

Note that the isometry group $\Isom(S_0)$ acts on $M$ and $S^{n-1}$ leaving the $x_{n-1-k}$-axis fixed and that $M$ contains a copy of $S^{n-1}\setminus S_\epsilon$. Extend the identity map of $S^{n-1}\setminus S_\epsilon$ to an $\Isom(S_0)$-equivariant diffeomorphism $\pi_0:M\to S^{n-1}$.
Let $C\subset M_{n-1-k}$ be the closed subset lying between the two copies of $S_0$ in $M$. Note that $C$ is a cylinder over $S^{n-2-k}$. Modify $\pi_0$ on a sufficiently small tubular neighborhood of $C$ to an $\Isom(S^{n-2-k})$-equivariant smooth map $\pi:M\to S^{n-1}$ such that $\pi$ restricts to the projection $C\twoheadrightarrow S^{n-2-k}$ and maps the complement of $C$ diffeomorphically onto that of $S^{n-2-k}$ preserving the stratification and the local parametrizations.

\medskip
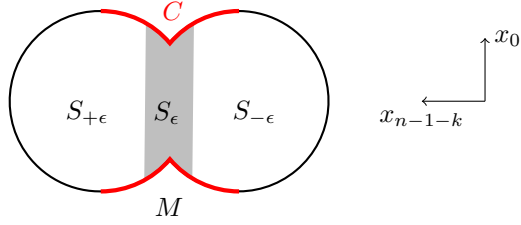
\begin{center}
\begin{tikzpicture}[scale=1.2]
\fill[color=lightgray] (-0.257,0.866) arc (60:40:1) arc (140:120:1) -- (0.257,-0.866) arc (240:220:1) arc (-40:-60:1) node[above=2.5em,right=0em][color=black]{$S_\epsilon$} node[above=2.5em,left=1em][color=black]{$S_{+\epsilon}$} node[above=2.5em,right=3em][color=black]{$S_{-\epsilon}$};
\draw[thick] (0,0.643) arc (40:320:1) arc (-140:140:1) node[below=5.5em]{$M$};
\draw[ultra thick,color=red] (-0.757,1) arc (90:40:1) arc (140:90:1) node[left=1.7em]{$C$};
\draw[ultra thick,color=red] (-0.757,-1) arc (-90:-40:1) arc (-140:-90:1);
\draw[->] (3.5,0) -- (3.5,.7) node[right]{$x_0$};
\draw[->] (3.5,0) -- (2.8,0) node[below]{$x_{n-1-k}$};
\end{tikzpicture}
\end{center}
\medskip

Now we are ready to define the composition $\CA_-\circ\CA_+:\CP\to\CQ$:
$$(\CA_-\circ\CA_+)(I) = \begin{cases}
 \CP(I), & \text{if $I\cap S^{n-1-k}_+=\emptyset$}, \\
 \CQ(I), & \text{if $I\cap S^{n-1-k}_-=\emptyset$}, \\
 \CO(\CH,\pi^{-1}(I)), & \text{if $I\cap S^{n-2-k}\ne\emptyset$}, \\
\end{cases}
$$
where $\CO(\CH,R)$ is the von Neumann algebra on $\CH$ generated by the images of $\CA(J)$ for all $R\supset J\in\CD$.

\begin{prop}
The pair $(\CA_-\circ\CA_+,\CH)$ is a well-defined $(k+1)$-morphism of $\widehat\Net^n$. 
\end{prop}
\begin{proof}[Sketch of proof]
We need to show that $(\CA_-\circ\CA_+,\CH)$ is a defect $n$-net of codimension $k+1$. The axioms of locality, additivity and covariance are clear.
The difficult part of the proof is to show the vacuum property (aka the $1\boxtimes1$-isomorphism in \cite{BDH19a}). 
%The idea of \cite{BDH19a} showing the vacuum property, aka the $1\boxtimes1$-isomorphism, does not work in higher dimensions. One has to take another strategy, for example, showing that the left module $\CH$ over $A:=(\CA_-\circ\CA_+)(S^{n-1}_\uparrow)$ is induced by a faithful semifinite normal weight.

Let $A = (\CA_-\circ\CA_+)(S^{n-1}_\uparrow)$. % and identify $A^\op$ with $(\CA_-\circ\CA_+)(S^{n-1}_\downarrow)$.
Identify $\CH_{\CA_\pm}$ with $L^2(\CA_\pm(S^{n-1}_\uparrow))$ so that it comes equipped with a modular conjugation $j_\pm:\CH_{\CA_\pm}\to\bar\CH_{\CA_\pm}$. Then $j_+$ and $j_-$ induce a conjugation $j:\CH\to\bar\CH$. % as well as an isomorphism $A\cong A^\op$, $a\mapsto j a^* j$.

Let $B = \CB(S_+)\vee\CA_+(S^{n-1}_\uparrow)$ be the von Neumann algebra on $\CH_{\CA_+}$. Choose a decomposition of the left $B$-module $\CH_{\CA_+} = \bigoplus \overline{B\xi_\alpha}$ where $\xi_\alpha\in\CH_{\CA_+}$ is $j_+$-invariant. 
%Then $\CH \cong \bigoplus \CH_{\CA_-} \boxtimes_{\CB(S_+)} \overline{B\xi_\alpha}$. 
Let $p_\alpha: \CH_{\CA_+} \to \overline{B\xi_\alpha}$ be the projection so that $p_\alpha\xi_\alpha=\xi_\alpha$ and $\sum p_\alpha = 1$. Since $p_\alpha \in B' = \CB(S_+)'\cap\CA_+(S^{n-1}_\downarrow)$, $j p_\alpha j$ defines a projection in $A$. 

Consider the trivial case $\CA_+=\Id_\CB$ so that $\CH_{\CA_+}\cong L^2(\CB(S_+))$. Recall that $\CH$ is a completion of $\hom_{\CB(S_-)}(L^2(\CB(S_-)),\CH_{\CA_-}) \otimes_{B(S_+)} \CH_{\CA_+}$.
Choose a decomposition of the left $A$-module $\CH = \bigoplus\overline{A(\phi_\beta\otimes\CH_{\CA_+})}$ where $\phi_\beta\in\hom_{\CB(S_-)}(L^2(\CB(S_-)),\CH_{\CA_-})$ is $j_-$-invariant. Let $q_\beta:\CH\to\overline{A(\phi_\beta\otimes\CH_{\CA_+})}$ be the projection so that $\sum q_\beta=1$. Then $q_\beta$ can be regarded as a projection in $\CB(S_-)'\cap\CA_-(S^{n-1}_\downarrow)$ such that $q_\beta\phi_\beta=\phi_\beta$.
%and $j q_\beta j$ defines a projection in $A$ orthogonal to $j p_\alpha j$.

For the general case, we have $\CH = \bigoplus\overline{A(\phi_\beta\otimes\xi_\alpha)}$. Let $\omega_{\alpha\beta}$ be the positive linear functional on $A$ associated to the $j$-invariant vector $\phi_\beta\otimes\xi_\alpha$. Then $(j p_\alpha q_\beta j)\omega_{\alpha\beta}=\omega_{\alpha\beta}$. Thus $\omega_{\alpha\beta}$ are mutually orthogonal.
Therefore, the faithful left $A$-module $\CH$ is induced by the weight $\sum\omega_{\alpha\beta}$. Since the weight $\sum\omega_{\alpha\beta}$ is compatible with $j$, it has to be faithful and the associated modular conjugation is $j$. This shows that the $A$-$A$-bimodule $\CH$ is isometric to $L^2(A)$, as desired.
\end{proof}

\begin{prop}
If both of $\CA_{\pm}$ belong to $\Net^n$, so is $\CA_-\circ\CA_+$.
\end{prop}
\begin{proof}[Sketch of proof]
%Let $r:S^{n-1}\to S^{n-1}$ be the reflection across the hyperplane $x_{n-1-k}=0$. Replacing $\CA_+$ and $\CA_-$ respectively by $\CA_+\oplus r^*\CA_-$ and $r^*\CA_+\oplus\CA_-$ if necessary, we may assume without loss of generality that $\CA_-=r^*\CA_+$ and $\CH_{\CA_-}=\bar\CH_{\CA_+}$.
We need to show that $\CA_-\circ\CA_+$ is finite.
The split property is clear. 
The duality condition is proved by constructing duality maps explicitly.
By constructing duality maps explicitly again, one shows that the $\CB(S_+)$-$\CO(\CH_{\CA_+},S_+')^\op$-bimodule $\CH_{\CA_+}$ is dualizable and similarly for the $\CO(\CH_{\CA_-},S_-')$-$\CB(S_+)$-bimodule $\CH_{\CA_-}$. Therefore, the $\CO(\CH_{\CA_-},S_-')$-$\CO(\CH_{\CA_+},S_+')^\op$-bimodule $\CH$ is dualizable. This implies that $\CH$ is a semisimple sector of $\CA_-\circ\CA_+$.
\end{proof}

\begin{rem}
The fusion of defects was carried out in \cite{BDH19a} by using (a variant of) fiber product of von Neumann algebras defined in \cite{Ti08}. However, the construction relies on the strong additivity axiom which is not satisfied by general topological nets, for example, those constructed in Section \ref{sec:net-oss} and Section \ref{sec:lw-net}. In the special case where $\CB$ does satisfy the strong additivity axiom, our construction of the fusion $\CA_-\circ\CA_+$ for $n=2$ clearly recovers that of \cite{BDH19a}.
\end{rem}

We have defined all ``vertical'' compositions of morphisms of $\widehat\Net^n$. ``Horizontal'' compositions are defined literally as ``vertical'' ones. It remains to construct various coherence relations of $\widehat\Net^n$ to establish Theorem \ref{thm:netn}. This part of work is essentially tautological but appeals to more sophisticated sewing operations of stratified nets.

So far, we have made a sufficient preparation for studying the condensation theory of topological nets which is the main topic of this paper. We shall leave the details of the construction of $\widehat\Net^n$ to a future work.

\begin{rem}
One can show that $\Net^n$ has duals: the dual of an object $(\CA,\CH_\CA)$ is $(\CA^\op,\bar\CH_\CA)$, the dual of an $n$-morphism $\CH$ is the complex conjugate $\bar\CH$ and the dual of a $k$-morphism $(\CF,\CH_\CF)$ for $1\le k<n$ is $(r_k^*\CF,\bar\CH_\CF)$ where $r_k$ is the reflection across the hyperplane $x_{n-k}=0$. The proof is an application of sewing arguments and will not be given here. See \cite{BDH19b} for a dualizability result of the symmetric monoidal $*$-3-category of conformal nets. We need not this fact in this paper.
\end{rem}

\section{Condensation theory of topological nets} \label{sec:condense-net}

The purpose of this section is to sketch a proof of Theorem \ref{thm:lqs-cc}.

\subsection{Condensation of 1D defects}

Let $\CR$ be a $*$-condensation monad on an $(n-1)$-morphism $\CA$ of $\widehat\Net^n$ (in particular, $\CR$ is a sector of $\CA$). Then $\CR$ induces a $*$-condensation monad on the von Neumann algebra $\CA(S^{n-1}_\uparrow)$ in the Morita $*$-2-category of von Neumann algebras. Applying Theorem \ref{thm:vn-cc} and Remark \ref{rem:vn-cc} then yields a canonical $*$-condensation $\CA\condense\CB$ extending the $*$-condensation monad $\CR$. More precisely, the $*$-condensation $\CA\condense\CB$ is defined by a pair of $n$-morphisms $\CH_\CB:\CA\to\CB$ and $\CR:\CB\to\CA$.

Now we assume that the $*$-condensation monad $\CR$ is defined in $\Net^n$. Since $n$-morphisms of $\Net^n$ are dualizable by Corollary \ref{cor:net-sec-dual}, we may assume that $\CR$ is unital (see \cite[Theorem 3.1.7]{GJF19}). According to Remark \ref{rem:vn-cc-dual}, $\CR:\CB\to\CA$ is right dual to $\CH_\CB:\CA\to\CB$.

It is easy to verify that $\CA\condense\CB$ is a well-defined $*$-condensation in $\Net^n$: Since $\CR:\CA\to\CA$ is semisimple and since the action of $\CA(S^{n-1}_\uparrow)$ on $\CR$ factors through $\CB(S^{n-1}_\uparrow)$, $\CR:\CB\to\CA$ is semisimple. Then $\CH_\CB:\CA\to\CB$ is also semisimple. Since the action of $\CA(S^{n-1}_\downarrow)$ on $\CH_\CB$ factors through $\CB(S^{n-1}_\downarrow)$, $\CH_\CB$ is a semisimple sector of $\CB$. Since $\Id_\CB$ is a direct summand of $\CH_\CB\circ\CR$, the duality condition of $\CA$ implies that of $\CB$. Hence $\CB$ is finite.

\subsection{Condensation of topological 2-nets}

We consider in this subsection the simplest nontrivial case $n=2$. Let $\CA=\underline{\C}$ be the trivial topological 2-net. We establish Theorem \ref{thm:lqs-cc} by showing that every $*$-condensation monad on $\CA$ admits a $*$-condensate in $\LQS^2$ and every $*$-condensation bimodule is induced by a 1-morphism of $\LQS^2$.

%Let $\CR$ be a unital $*$-condensation monad on an object $\CA\in\LQS^2$. We construct explicitly a $*$-condensate $\CB$ of $\CR$. 
%Since $\LQS^2$ has duals, we may assume that $\CR$ is unital (see \cite[Theorem 3.1.7]{GJF19}). 
Let $\CR$ be a unital $*$-condensation monad on $\CA$.
To unpack the definition of $\CR$, we assume a prior that $\CR$ is induced by a $*$-condensation $f:\CA\condense\CE$ defined by the counit map $v:f\circ f^\vee\to\Id_\CE$ and a 3-morphism $\alpha:v\circ v^\vee\to\Id_{\Id_\CE}$ such that $\alpha\circ\alpha^*=1$. (Indeed, $\CE$ lives in the $*$-condensation completion of $\LQS^2$ which is equivalent to $3\Hilb$.) See the following picture. By composing the 1-morphisms in the left configuration horizontally to eliminate $\CE$, we obtain a configuration in terms of the defining data of $\CR$ on the right.
\medskip
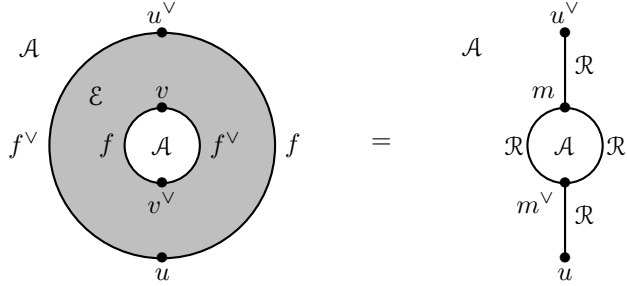
\begin{center}
\begin{tikzpicture}[scale=1]
\fill[lightgray] (0,0) circle (1.5);
\draw[thick] (1.5,0) node[right]{$f$} 
  arc(0:90:1.5) node{$\bullet$} node[above]{$u^\vee$}
  arc(90:180:1.5) node[left]{$f^\vee$}
  arc(180:270:1.5) node{$\bullet$} node[below]{$u$}
  arc(270:360:1.5);
\fill[fill=white] (0,0) circle (0.5) node{$\CA$} node[left=5em,above=3em]{$\CA$} node[left=2.5em,above=1.2em]{$\CE$};
\draw[thick] (.5,0) node[right]{$f^\vee$} 
  arc(0:90:.5) node{$\bullet$} node[above]{$v$}
  arc(90:180:.5) node[left]{$f$}
  arc(180:270:.5) node{$\bullet$} node[below]{$v^\vee$}
  arc(270:360:.5);
\end{tikzpicture}
\quad\quad \raisebox{5.3em}{$=$} \quad\quad
\begin{tikzpicture}[scale=1]
\draw[thick] (0,0.5) node{$\bullet$} node[above left]{$m$} 
  -- (0,1.5) node{$\bullet$} node[above]{$u^\vee$} node[below=1.2em,right]{$\CR$};
\draw[thick] (0,-1.5) node{$\bullet$} node[below]{$u$} 
  -- (0,-0.5) node{$\bullet$} node[below left]{$m^\vee$} node[below=1.2em,right]{$\CR$};
\draw[thick] (0,0) circle (0.5) node{$\CA$} node[left=3.5em,above=3em]{$\CA$}
 node[left=1.2em]{$\CR$} node[right=1.2em]{$\CR$};
\end{tikzpicture}
\end{center}
\medskip
The 2-morphisms $u:\Id_\CA\to\CR$ and $m:\CR\circ\CR\to\CR$ exhibit $\CR$ as an algebra in the monoidal $*$-2-category $\Hom_{\LQS^2}(\CA,\CA)$. Moreover, the 2-morphism $m^\vee\circ u:\Id_\CA\to\CR\circ\CR$ exhibits $\CR$ self dual.

Let $\theta:\Id_\CA\to\CR$ be a 2-morphism which contains all the simple 2-morphisms $\Id_\CA\to\CR$ as direct summands and admits an isometric embedding $\mu:u\hookrightarrow\theta$. Note that $\theta^\vee:\CR\to\Id_\CA$ induces a 2-morphism $\Id_\CA\to\CR^\vee$ which by the self duality of $\CR$ determines a 2-morphism $\theta^*:\Id_\CA\to\CR$.

We follow the lines of Section \ref{sec:lw-net} to construct a $*$-condensate $\CB$ of $\CR$.
Let $P=\{p_1,\dots,p_k\}\subset S^1$ be a finite subset where the points $p_i$ are in cyclic order. We have a sector of $\CA$ defined by the following composition:
$$\CH_P: \Id_\CA \xrightarrow{(\theta\theta^*)^k} \CR^{2k} \xrightarrow{m} \CR \xrightarrow{u^\vee} \Id_\CA$$
where $\CR^k$ denotes the $k$-fold composition $\CR\circ\CR\circ\cdots\circ\CR$. For any inclusion $P\subset Q$, the isometric embedding $\mu:u\hookrightarrow\theta$ induces an isometric embedding $\CH_P\hookrightarrow\CH_Q$. Form a direct limit in $\Sect(\CA)$
$$\CH_\CB = \varinjlim_P \CH_P.$$

To facilitate the generalization to higher dimensions, we give a diagrammatic description of $\CH_P$ by means of the postulated $*$-condensate $\CE$. Consider the following configuration where we insert a copy of the 2-morphism $\theta$ (resp. $\theta^*$) on the left (resp. right) of each point of $P$:
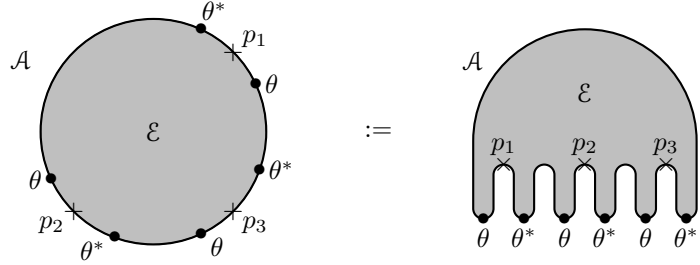
\medskip
\begin{center}
\begin{tikzpicture}[scale=1]
\filldraw[thick,fill=lightgray] (0,0) circle (1.5) node[left=5em,above=2em]{$\CA$} node[]{$\CE$};
\draw (1.5,0) 
 arc (0:25:1.5) node{$\bullet$} node[right]{$\theta$}
 arc (25:45:1.5) node{$+$} node[above=.5em,right]{$p_1$}
 arc (45:65:1.5) node{$\bullet$} node[right=.5em,above]{$\theta^*$}
 arc (65:205:1.5) node{$\bullet$} node[left]{$\theta$}
 arc (205:225:1.5) node{$+$} node[below=.5em,left]{$p_2$}
 arc (225:250:1.5) node{$\bullet$} node[below=.5em,left]{$\theta^*$}
 arc (250:295:1.5) node{$\bullet$} node[below=.5em,right]{$\theta$}
 arc (295:315:1.5) node{$+$} node[below=.5em,right]{$p_3$}
 arc (315:340:1.5) node{$\bullet$} node[right]{$\theta^*$}
;
\end{tikzpicture}
\quad\quad \raisebox{5em}{$:=$} \quad\quad
\raisebox{.5em}{\begin{tikzpicture}[scale=.9]
\filldraw[thick,fill=lightgray] (1.65,0) arc (0:180:1.65 and 1.65) node[above=2.3em]{$\CA$}
 -- (-1.65,-1) arc (-180:-90:.15) node{$\bullet$} node[below]{$\theta$} arc (-90:0:.15)  
 -- (-1.35,-.5) arc (180:90:.15) node{$\times$} node[above]{$p_1$} arc (90:0:.15)
 -- (-1.05,-1) arc (-180:-90:.15) node{$\bullet$} node[below]{$\theta^*$} arc (-90:0:.15)
 -- (-.75,-.5) arc (180:0:.15) 
 -- (-.45,-1) arc (-180:-90:.15) node{$\bullet$} node[below]{$\theta$} arc (-90:0:.15)
 -- (-.15,-.5) arc (180:90:.15) node{$\times$} node[above]{$p_2$} node[above=2em]{$\CE$} arc (90:0:.15)
 -- (.15,-1) arc (-180:-90:.15) node{$\bullet$} node[below]{$\theta^*$} arc (-90:0:.15)
 -- (.45,-.5) arc (180:0:.15) 
 -- (.75,-1) arc (-180:-90:.15) node{$\bullet$} node[below]{$\theta$} arc (-90:0:.15)
 -- (1.05,-.5) arc (180:90:.15) node{$\times$} node[above]{$p_3$} arc (90:0:.15)
 -- (1.35,-1) arc (-180:-90:.15) node{$\bullet$} node[below]{$\theta^*$} arc (-90:0:.15)
 -- (1.65,0);
\end{tikzpicture}}
\end{center}
\medskip
The left diagram is interpreted by the right one. By composing the 1-morphisms horizontally to eliminate $\CE$, we obtain a configuration of which the vertical composition recovers the above definition of $\CH_P$.

%\begin{rem}
%Note that a finite subset $P\subset S^1$ can be regarded as a cell decomposition of $S^1$ of which the set of vertices is $P$. Let $\hat P$ be the barycentric subdivision of $P$, i.e. one inserts a vertex in to each edge of $P$, and let $\hat P^\vee$ be the dual decomposition of $\hat P$, i.e. the set of vertices of $\hat P^\vee$ is in one-to-one correspondence to that of $\hat P$. 
%\end{rem}

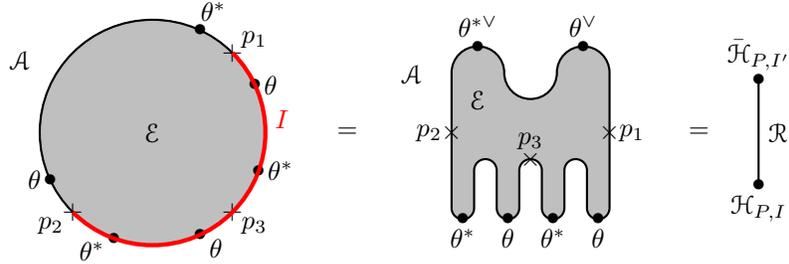
\medskip
\begin{center}
\begin{tikzpicture}[scale=1]
\filldraw[thick,fill=lightgray] (0,0) circle (1.5) node[left=5em,above=2em]{$\CA$} node[]{$\CE$};
\draw (1.5,0) 
 arc (0:25:1.5) node{$\bullet$} node[right]{$\theta$}
 arc (25:45:1.5) node{$+$} node[above=.5em,right]{$p_1$}
 arc (45:65:1.5) node{$\bullet$} node[right=.5em,above]{$\theta^*$}
 arc (65:205:1.5) node{$\bullet$} node[left]{$\theta$}
 arc (205:225:1.5) node{$+$} node[below=.5em,left]{$p_2$}
 arc (225:250:1.5) node{$\bullet$} node[below=.5em,left]{$\theta^*$}
 arc (250:295:1.5) node{$\bullet$} node[below=.5em,right]{$\theta$}
 arc (295:315:1.5) node{$+$} node[below=.5em,right]{$p_3$}
 arc (315:340:1.5) node{$\bullet$} node[right]{$\theta^*$}
;
\draw[ultra thick,color=red] (1.061,1.061) node[below=2.5em,right=1.2em]{$I$} arc (45:-135:1.5);
\end{tikzpicture}
\quad \raisebox{5em}{$=$} \quad
\raisebox{.5em}{\begin{tikzpicture}[scale=1]
\filldraw[thick,fill=lightgray] 
 (-1.05,0) node{$\times$} node[left]{$p_2$} node[left=1.5em,above=1.5em]{$\CA$} node[right=1em,above=.5em]{$\CE$}
 -- (-1.05,-1) arc (-180:-90:.15) node{$\bullet$} node[below]{$\theta^*$} arc (-90:0:.15)
 -- (-.75,-.5) arc (180:0:.15) 
 -- (-.45,-1) arc (-180:-90:.15) node{$\bullet$} node[below]{$\theta$} arc (-90:0:.15)
 -- (-.15,-.5) arc (180:90:.15) node{$\times$} node[above]{$p_3$} arc (90:0:.15)
 -- (.15,-1) arc (-180:-90:.15) node{$\bullet$} node[below]{$\theta^*$} arc (-90:0:.15)
 -- (.45,-.5) arc (180:0:.15) 
 -- (.75,-1) arc (-180:-90:.15) node{$\bullet$} node[below]{$\theta$} arc (-90:0:.15)
 -- (1.05,0) node{$\times$} node[right]{$p_1$}
 -- (1.05,.8) arc (0:90:.35) node{$\bullet$} node[above]{$\theta^\vee$} arc (90:180:.35)
 arc (0:-180:.35)
 arc (0:90:.35) node{$\bullet$} node[above]{$\theta^{*\vee}$} arc (90:180:.35)
 -- (-1.05,0);
\end{tikzpicture}}
\quad \raisebox{5em}{$=$} 
\raisebox{1.5em}{\begin{tikzpicture}[scale=1]
\draw[thick] (0,-.7) node{$\bullet$} node[below]{$\CH_{P,I}$} -- (0,0) node[right]{$\CR$} -- (0,.7) node{$\bullet$} node[above]{$\bar\CH_{P,I'}$};
\end{tikzpicture}}
\end{center}
\medskip

Let $I\subset S^1$ be an arc and let $P\subset S^1$ be a finite subset. Enlarging $P$ if necessary, we assume $\partial I\subset P$. Note that $I$ and $I'$ divide the collection of the inserted 2-morphisms $\theta$ and $\theta^*$ into two disjoint parts. Correspondingly, $\CH_P$ admits a decomposition as depicted in the above picture
$$\CH_P: \Id_\CA \xrightarrow{\CH_{P,I}} \CR \xrightarrow{\bar\CH_{P,I'}} \Id_\CA.$$
Passing to the direct limit yields a decomposition
$$\CH_\CB: \Id_\CA \xrightarrow{\CH_{\CB,I}} \CR \xrightarrow{\bar\CH_{\CB,I'}} \Id_\CA.$$
Define $\CB(I)$ to be the von Neumann algebra
$$\CB(I) = \hom_{\CR(S^1_\uparrow)}(\CH_{\CB,I},\CH_{\CB,I}).$$

We need to verify that $(\CB,\CH_\CB)$ is a finite topological 2-net. Note that $\bar\CH_{P,I}\boxtimes_{\CR(S^1_\uparrow)}\CH_{P,I}$ defines a unital $*$-condensation monad on $\Id_\CA$, determining a $*$-condensation $\Id_\CA\condense\tilde\CB$ in $\LQS^2$ with $\tilde\CB(S^1_\uparrow) = \hom_{\CR(S^1_\uparrow)}(\CH_{P,I},\CH_{P,I})$ and $\CH_{\tilde\CB} = \bar\CH_{P,I}\boxtimes_{\CR(S^1_\uparrow)}\CH_{P,I}$. In particular, $\CH_{\tilde\CB} \cong L^2(\tilde\CB(S^1_\uparrow))$. Passing to the direct limit yields $\CH_\CB \cong L^2(\CB(S^1_\uparrow))$. This proves the vacuum property of $\CB$. The center $Z(\CO(\CH_\CB))$ can be computed by any finite configuration. Hence $\CH_\CB$ is a semisimple sector of $\CB$. The duality condition of $\CB$ is also proved by reducing the problem to finite configurations.

In order to obtain a 1-morphism $\CA\to\CB$ to extend the $*$-condensation monad $\CR$, one simply applies the above construction on a half circle instead of $S^1$.

\begin{rem}
%In the special case where $\CA=\underline{\C}$, $\CR$ is a $*$-condensation algebra in $\Omega\LQS^2 = \LQS^1 \simeq 2\Hilb$. Thus 
Since $\CR$ is a unital $*$-condensation algebra in $\LQS^1 \simeq 2\Hilb$, 
$\CR$ may be regarded as a unitary multi-fusion 1-category $\CC$ so that $m:\CR\circ\CR\to\CR$ is identified with the tensor product $\otimes:\CC\boxtimes\CC\to\CC$ and $u:\Id_\CA\to\CR$ defines the tensor unit of $\CC$. Moreover, the 1-morphism $\theta: \Id_\CA \to \CR$ of $\LQS^1$ defines an object $V\in\CC$. The topological 2-net $(\CB,\CH_\CB)$ constructed above recovers the Levin-Wen net from Section \ref{sec:lw-net} associated to $\CC$ and $V$.
\end{rem}

Now we consider a $*$-condensation bimodule $\CS$ over unital $*$-condensation monads $\CR$ and $\CR'$ on $\CA$. To unpack the definition of $\CS$, we assume a prior that $\CR$ and $\CR'$ are induced by $*$-condensations $f:\CA\condense\CE$ and $f':\CA\condense\CE'$, respectively, and $\CS=f'^\vee\circ g\circ f$ where $g:\CE\to\CE'$ is a 1-morphism of $\LQS^2$. See the following picture:
\medskip
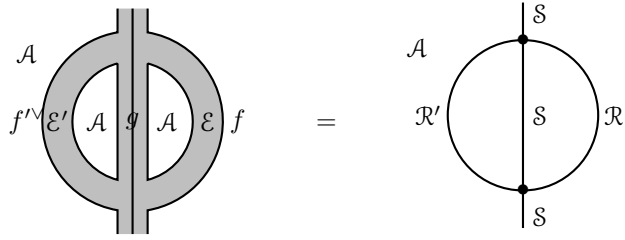
\begin{center}
\begin{tikzpicture}[scale=1]
\fill[fill=lightgray] (-.2,1.5) -- (.2,1.5) -- (.2,1.183) arc (80.4:-80.4:1.2) -- (.2,-1.5) 
  -- (-.2,-1.5) -- (-.2,-1.183) arc (260.4:99.6:1.2) -- (-.2,1.5);
\draw[thick] (.2,1.5) -- (.2,1.183) arc (80.4:-80.4:1.2) -- (.2,-1.5); 
\draw[thick] (-.2,-1.5) -- (-.2,-1.183) arc (260.4:99.6:1.2) -- (-.2,1.5);
\draw[thick] (0,1.5) -- (0,0) node{$g$} -- (0,-1.5);
\node at (-1.4,.9) {$\CA$}; 
\node at (-1,0) {$\CE'$};
\node at (1,0) {$\CE$};
\node at (-1.4,0) {$f'^\vee$};
\node at (1.4,0) {$f$};
\filldraw[thick,fill=white] (.2,.77) arc(75.5:-75.5:.8) -- (.2,0) node[right]{$\CA$} -- (.2,.77);
\filldraw[thick,fill=white] (-.2,.77) arc(104.5:255.5:.8) -- (-.2,0) node[left]{$\CA$} -- (-.2,.77);
\end{tikzpicture}
\quad\quad \raisebox{4em}{$=$} \quad\quad
\begin{tikzpicture}[scale=1]
\draw[thick] (0,1.5) -- (0,1) node[above=1em,right]{$\CS$} node{$\bullet$} 
  -- (0,0) node[right]{$\CS$} 
  -- (0,-1) node[below=1em,right]{$\CS$} node{$\bullet$} 
  -- (0,-1.5) ;
\draw[thick] (0,0) circle (1) node[left=2.7em]{$\CR'$} node[right=2.7em]{$\CR$};
\node at (-1.4,.9) {$\CA$}; 
\end{tikzpicture}
\end{center}
\medskip

Let $\zeta:\Id_\CA\to\CS$ be a 2-morphism which contains all the simple 2-morphisms $\Id_\CA\to\CS$ as direct summands. Let $I\in\Disk^1_1$ be an arc and let $P\subset S^1$ be a finite subset containing $\partial I$. Assume without loss of generality that $I$ intersects $S^0_\downarrow$. We have a sector $\CH_P$ of $\CA$ as depicted in the following picture:
\medskip
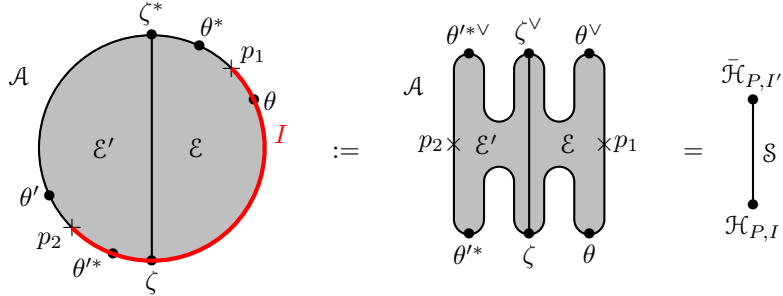
\begin{center}
\begin{tikzpicture}[scale=1]
\filldraw[thick,fill=lightgray] (0,0) circle (1.5) node[left=5em,above=2em]{$\CA$} 
  node[left=1em]{$\CE'$} node[right=1em]{$\CE$};
\draw (1.5,0) 
 arc (0:25:1.5) node{$\bullet$} node[right]{$\theta$}
 arc (25:45:1.5) node{$+$} node[above=.5em,right]{$p_1$}
 arc (45:65:1.5) node{$\bullet$} node[right=.5em,above]{$\theta^*$}
 arc (65:205:1.5) node{$\bullet$} node[left]{$\theta'$}
 arc (205:225:1.5) node{$+$} node[below=.5em,left]{$p_2$}
 arc (225:250:1.5) node{$\bullet$} node[below=.5em,left]{$\theta'^*$}
;
\draw[thick] (0,1.5) node[above]{$\zeta^*$} node{$\bullet$} -- (0,0) -- (0,-1.5) node[below]{$\zeta$} node{$\bullet$};
\draw[ultra thick,color=red] (1.061,1.061) node[below=2.5em,right=1.2em]{$I$} arc (45:-135:1.5);
\end{tikzpicture}
\quad \raisebox{5.5em}{$:=$} \quad
\raisebox{1em}{\begin{tikzpicture}[scale=1]
\filldraw[thick,fill=lightgray] 
 (-1,0) node{$\times$} node[left]{$p_2$} node[left=1.5em,above=1.5em]{$\CA$}
 -- (-1,-1) arc (-180:-90:.2) node{$\bullet$} node[below]{$\theta'^*$} arc (-90:0:.2)
 -- (-.6,-.5) arc (180:0:.2) 
 -- (-.2,-1) arc (-180:-90:.2) node{$\bullet$} node[below]{$\zeta$} arc (-90:0:.2)
 -- (.2,-.5) arc (180:0:.2) 
 -- (.6,-1) arc (-180:-90:.2) node{$\bullet$} node[below]{$\theta$} arc (-90:0:.2)
 -- (1,0) node{$\times$} node[right]{$p_1$}
 -- (1,1) arc (0:90:.2) node{$\bullet$} node[above]{$\theta^\vee$} arc (90:180:.2)
 -- (.6,.5) arc (0:-180:.2) 
 -- (.2,1) arc (0:90:.2) node{$\bullet$} node[above]{$\zeta^\vee$} arc (90:180:.2)
 -- (-.2,.5) arc (0:-180:.2) 
 -- (-.6,1) arc (0:90:.2) node{$\bullet$} node[above]{$\theta'^{*\vee}$} arc (90:180:.2)
 -- (-1,0);
\draw[thick] (0,1.2) -- (0,0) node[left=.8em]{$\CE'$} node[right=.8em]{$\CE$} -- (0,-1.2);
\end{tikzpicture}}
\quad \raisebox{5.5em}{$=$} 
\raisebox{2em}{\begin{tikzpicture}[scale=1]
\draw[thick] (0,-.7) node{$\bullet$} node[below]{$\CH_{P,I}$} -- (0,0) node[right]{$\CS$} -- (0,.7) node{$\bullet$} node[above]{$\bar\CH_{P,I'}$};
\end{tikzpicture}}
\end{center}
\medskip
Form direct limits in $\Sect(\CA)$
$$\CH_\CF = \varinjlim_P \CH_P, \quad \CH_{\CF,I} = \varinjlim_P \CH_{P,I}$$
and define $\CF(I)$ to be the von Neumann algebra
$$\CF(I) = \hom_{\CS(S^1_\uparrow)}(\CH_{\CF,I},\CH_{\CF,I}).$$
Then $(\CF,\CH_\CF)$ defines a 1-morphism of $\LQS^2$ inducing the $*$-condensation bimodule $\CS$.

\subsection{Condensation of topological $n$-nets} \label{sec:cond-nnet}

Let $\CA=\underline{\C}$ be the trivial topological $n$-net where $n\ge2$. By induction on $n$, we have $\Hom_{\LQS^n}(\CA,\CA) \simeq n\Hilb$.

Let $\CR$ be a unital $*$-condensation monad on $\CA$. We construct explicitly a $*$-condensate $\CB$ of $\CR$ by generalizing the construction from the previous subsection. 
To unpack the definition of $\CR$, we assume a prior that $\CR$ is induced by a $*$-condensation $f:\CA\condense\CE$.
% defined by the consecutive counit maps $v_2:f\circ f^\vee\to\Id_\CE$, $v_3:v_2\circ v_2^\vee\to\Id_{\Id_\CE}$, etc. terminated by an $(n+1)$-morphism $\alpha:v_n\circ v_n^\vee\to\Id_{\cdots\Id_\CE}$ such that $\alpha\circ\alpha^*=1$. %There are a pair of 2-morphism $u:\Id_\CA\to\CR$ and $m:\CR\circ\CR\to\CR$ exhibiting $\CR$ as an algebra in the monoidal $*$-$n$-category $\Hom_{\LQS^n}(\CA,\CA)$.

Consider the cone $\Lambda = \{ x\in\R^n \mid x_i\ge0 \}$. We label the interior of $\Lambda$ by $\CE$ and the complement of $\Lambda$ by $\CA$. Use $\theta_1$ to denote the 1-morphism $f:\CA\to\CE$ and assign it to all the faces of $\Lambda$ of codimension one. Then by induction on $k$ for $2\le k\le n$, we choose a $k$-morphism $\theta_k:\Id_{\cdots\Id_\CA}\to f_k$ and assign it to all the faces of codimension $k$, where $f_k$ is obtained by composing the morphisms around such a face (for example, $f_2=f^\vee\circ f=\CR$). We require that $\theta_k$ extend to a $*$-condensation and contain the unit map $u_{\theta_{k-1}}: \Id_{\cdots\Id_\CA} \to \theta_{k-1}^\vee\circ\theta_{k-1}$ as a direct summand. Moreover, let $\theta_n^*:\Id_{\cdots\Id_\CA}\to f_n$ be the $n$-morphism induced by $\theta_n$ due to the self duality of $\CR$. Note that the morphisms $\theta_2,\dots,\theta_n$ and $\theta_n^*$ are irrelevant to $\CE$.

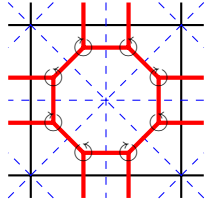
\medskip
\begin{center}
\begin{tikzpicture}[scale=1]
\draw[thick] (-1.3,1)--(1.3,1)  (-1.3,-1)--(1.3,-1)  (1,-1.3)--(1,1.3)  (-1,-1.3)--(-1,1.3);
\draw[dashed,color=blue] (-1.3,-1.3) -- (1.3,1.3)  (1.3,-1.3) -- (-1.3,1.3)  (0,1.3) -- (0,-1.3)  (1.3,0) -- (-1.3,0)
  (.7,1.3)--(1.3,.7)  (.7,-1.3)--(1.3,-.7)  (-.7,1.3)--(-1.3,.7)  (-.7,-1.3)--(-1.3,-.7);
\draw[ultra thick,color=red] 
  (.3,.7) node[color=black]{$\circlearrowleft$} -- (.7,.3) node[color=black]{$\circlearrowright$}
  -- (.7,-.3) node[color=black]{$\circlearrowleft$} -- (.3,-.7) node[color=black]{$\circlearrowright$}
  -- (-.3,-.7) node[color=black]{$\circlearrowleft$} -- (-.7,-.3) node[color=black]{$\circlearrowright$}
  -- (-.7,.3) node[color=black]{$\circlearrowleft$} -- (-.3,.7) node[color=black]{$\circlearrowright$} -- (.3,.7) 
  (-.3,.7)--(-.3,1.3)  (.3,.7)--(.3,1.3)  (-.3,-.7)--(-.3,-1.3)  (.3,-.7)--(.3,-1.3)  
  (.7,.3)--(1.3,.3)  (.7,-.3)--(1.3,-.3)  (-.7,.3)--(-1.3,.3)  (-.7,-.3)--(-1.3,-.3) ;
\end{tikzpicture}
\end{center}
\medskip

Let $P$ be a regular cell decomposition of $S^{n-1}$ (compatible with the smooth structure). Let $\hat P$ be a barycentric subdivision of $P$ and let $\hat P^\vee$ be a dual decomposition of $\hat P$. In particular, the vertices of $\hat P$ are in bijection to all the cells of $P$ and the $i$-cells of $\hat P^\vee$ are in bijection to the $(n-1-i)$-cells of $\hat P$. 
See the above picture, where $P$ is depicted by the black lines and $\hat P^\vee$ is depicted by the red lines.
Indeed, $\hat P$ is a triangulation of $S^{n-1}$ and each cell of $\hat P$ carries a canonical orientation (determined by the evident linear order on its vertices). Then, each vertex of $\hat P^\vee$ carries an induced orientation.

Assign the $k$-morphism $\theta_k$ to each $(n-k)$-cell of $\hat P^\vee$ for $1\le k<n$ and assign the $n$-morphism $\theta_n$ or $\theta_n^*$ to each vertex according to the orientation. Composing all the morphisms in the configuration then yields a sector $\CH_P$ of $\CA$ which is irrelevant to $\CE$. Since $\hat P$ and $\hat P^\vee$ are unique up to isotopy, $\CH_P$ is independent of the choice of them. 
If $Q$ is a subdivision of $P$, then $\hat Q^\vee$ is canonically a subdivision of $\hat P^\vee$ up to isotopy. The inclusions $u_{\theta_k}\hookrightarrow\theta_{k+1}$ then induce an isometric embedding $\CH_P\hookrightarrow\CH_Q$. Form a direct limit in $\Sect(\CA)$
$$\CH_\CB = \varinjlim_P \CH_P.$$

Let $I\subset S^{n-1}$ be a disk region and let $P$ be a regular cell decomposition of $S^{n-1}$. Passing to a subdivision of $P$ if necessary, we assume that $\partial I$ is a union of cells of $P$ so that $\partial I$ is transverse to all the cells of $\hat P^\vee$. Then $\CH_P$ admits a decomposition
$$\CH_P: \Id_{\cdots\Id_\CA} \xrightarrow{\CH_{P,I}} \CT_{P,I} \xrightarrow{\bar\CH_{P,I'}} \Id_{\cdots\Id_\CA}$$
where $\CT_{P,I}$ is obtained by composing all the morphisms lying on $\partial I$ and $\CH_{P,I}$ is obtained by composing all the morphisms inside $I$.
Passing to the direct limit yields a decomposition
$$\CH_\CB: \Id_{\cdots\Id_\CA} \xrightarrow{\CH_{\CB,I}} \CT_I \xrightarrow{\bar\CH_{\CB,I'}} \Id_{\cdots\Id_\CA}.$$
Define $\CB(I)$ to be the von Neumann algebra
$$\CB(I) = \hom_{\CT_I(S^{n-1}_\uparrow)}(\CH_{\CB,I},\CH_{\CB,I}).$$
Then $(\CB,\CH_\CB)$ defines a $*$-condensate of the $*$-condensation monad $\CR$.

\begin{rem}
Since $\CR$ is a $*$-condensation algebra in $\LQS^{n-1} \simeq n\Hilb$,
%In the special case where $\CA=\underline{\C}$, $\CR$ is a $*$-condensation algebra in $\Omega\LQS^n = \LQS^{n-1} \simeq n\Hilb$. Thus 
$\CR$ may be regarded as a unitary multi-fusion $(n-1)$-category $\CC$. Then $\theta_2$ defines an object of $\CC$ and $\theta_k$ defines a $(k-2)$-morphism of $\CC$ for $3\le k\le n$. 
The construction of the topological $n$-net $(\CB,\CH_\CB)$ can be reduced to a computation in the unitary multi-fusion $(n-1)$-category $\CC$ without referring to the local observable algebras of $\CR$. Moreover, $\CT_I$ lives in the completion of $\Omega^{n-1}\LQS^n = \LQS^1$ by infinite direct sums therefore $\CB(I)$ is an atomic type I von Neumann algebra.
\end{rem}

Generalizing the above construction, we see that all $*$-condensation bimodules and bimodule maps, etc. are induced by morphisms of $\LQS^n$. This completes the proof of Theorem \ref{thm:lqs-cc}.

\appendix
\section{Condensation theory of von Neumann algebras} \label{sec:condense-vn}

Recall that functors between $*$-$n$-categories are silently assumed to be $*$-functors.

\subsection{Lurie's formulation of Connes fusion}

We review in this subsection Lurie's formulation of Connes fusion of bimodules \cite[Lecture 21]{Lur11}.

\smallskip

Let $A$ and $B$ be von Neumann algebras. We use $\BMod_{A|B}$ to denote the $*$-1-category of $A$-$B$-bimodules and use $\LMod_A$ to denote the $*$-1-category of left $A$-modules.
We remind the reader the following basic fact and give a proof for the reader's convenience. It is analogous to the fact that every left module over a semisimple algebra is a direct summand of a free module.

\begin{lem} \label{lem:lmod-l2}
Every left $A$-module is a direct summand of a direct sum of (possibly infinite) copies of $L^2(A)$.
\end{lem}
\begin{proof}
By Zorn's lemma, every left $A$-module is a direct sum of cyclic left $A$-modules (see \cite[Proposition I.9.17]{Ta79}). A cyclic left $A$-module induced by a normal state $\omega$ is a direct summand of $L^2(A)$ by \cite[Lemma IX.1.6]{Ta03}.
%If $M$ is a cyclic left $A$-module associated to a normal state $\psi$, then $M\cong L^2(A)s(\psi)$ where $s(\psi)$ is the support of $\psi$. In particular, $M$ is a direct summand of $L^2(A)$.
\end{proof}

We say that a functor $F:\LMod_A\to\LMod_B$ is {\em completely additive} if it satisfies the following condition:
\begin{itemize}
\item For every (possibly infinite) collection of objects $M_\alpha$ in $\LMod_A$, the collection of maps $F(M_\alpha)\to F(\bigoplus M_\alpha)$ exhibits $F(\bigoplus M_\alpha)$ as the direct sum $\bigoplus F(M_\alpha)$.
\end{itemize}
We use $\Fun^c(\LMod_A,\LMod_B)$ to denote the $*$-1-category of completely additive functors $F:\LMod_A\to\LMod_B$.

\begin{lem} \label{lem:cc-normal}
Let $F:\LMod_A\to\LMod_B$ be a completely additive functor. Then for every $M\in\LMod_A$, the map $$F_M: \Hom_{\LMod_A}(M,M) \to \Hom_{\LMod_B}(F(M),F(M))$$ is a (normal) homomorphism of von Neumann algebras.
\end{lem}
\begin{proof}
According to \cite[Corollary III.3.11]{Ta79}, a bounded linear functional $\omega:R\to\C$ on a von Neumann algebra $R$ is normal if and only if $\omega(\sum e_\alpha)=\sum\omega(e_\alpha)$ for every orthogonal family $e_\alpha$ of projections in $R$. Moreover, giving an orthogonal family of projections in $\Hom_{\LMod_A}(M,M)$ is equivalent to giving an orthogonal family of direct summands of the left $A$-module $M$. Therefore, the complete additivity of $F$ implies that if $\omega$ is a normal linear functional on the target of $F_M$ then the induced functional $\omega\circ F_M$ on the domain is also normal. That is, $F_M$ is normal.
\end{proof}

\begin{prop} \label{prop:nfun-bimod}
We have a pair of mutually inverse functors:
$$\Fun^c(\LMod_A,\LMod_B) \to \BMod_{B|A}, \quad F\mapsto F(L^2(A)),$$
$$\BMod_{B|A} \to \Fun^c(\LMod_A,\LMod_B) , \quad M\mapsto M\boxtimes_A-.$$
\end{prop}
\begin{proof}
The proof is straightforward: $F(L^2(A))$ is a well-defined $B$-$A$-bimodule by Lemma \ref{lem:cc-normal}. On the one hand, $M\boxtimes_A L^2(A) \cong M$. 
On the other hand, by Lemma \ref{lem:lmod-l2} $F$ is determined by $F(L^2(A))$ up to canonical natural isomorphism. Namely, $F\cong F(L^2(A))\boxtimes_A-$.
\end{proof}

The following corollary states that Connes fusion of bimodules can be formulated by composition of completely additive functors.

\begin{cor}
The bifunctor of Connes fusion
$$\BMod_{C|B} \times \BMod_{B|A} \to \BMod_{C|A}, \quad (M,N)\mapsto M\boxtimes_B N$$
is induced by the bifunctor of composition 
$$\Fun^c(\LMod_B,\LMod_C) \times \Fun^c(\LMod_A,\LMod_B) \to \Fun^c(\LMod_A,\LMod_C),$$
$$(G,F)\mapsto G\circ F.$$
\end{cor}

\subsection{Condensation of von Neumann algebras}

Let $\Mor(\VN)$ denote the $*$-2-category of von Neumann algebras and bimodules \cite{Lan01}.

\begin{thm} \label{thm:vn-cc}
$\Mor(\VN)$ is $*$-condensation-complete.
\end{thm}

\begin{proof}
Let $R$ be a $*$-condensation monad on $A\in\Mor(\VN)$. That is, $R$ is a nonunital algebra in $\BMod_{A|A}$ such that $m^*$ is an $R$-$R$-bimodule map and $m\circ m^* = \Id_R$ where $m:R\boxtimes_A R\to R$ is the multiplication of $R$. We need to show that $R$ is induced by a $*$-condensation $A\condense B$. That is, there exist a pair of bimodules ${}_A M_B$ and ${}_B N_A$, a bimodule map $t:N\boxtimes_A M\to L^2(B)$ such that $t\circ t^*=\Id_{L^2(B)}$, and a bimodule isomorphism $M\boxtimes_B N\cong R$ rendering the following diagram commutative:
$$\xymatrix{
  M\boxtimes_B N\boxtimes_A M\boxtimes_B N \ar[r]^-\sim \ar[d]_{M\boxtimes_B t\boxtimes_B N} & R\boxtimes_A R \ar[d]^m \\
  M\boxtimes_B N \ar[r]^\sim & R .
}
$$

Let $T$ denote the nonunital algebra in $\Fun^c(\LMod_A,\LMod_A)$ induced by the $*$-condensation monad $R$. Define a $*$-1-category $\CC$ whose objects are the left $T$-modules $X$ in $\LMod_A$ such that $t_X^*$ is a left $T$-module map and $t_X\circ t_X^*=\Id_X$ where $t_X:T(X)\to X$ is the action map and whose morphisms are the left $T$-module maps. 
Then $\Hom_\CC(X,X)$ is a von Neumann algebra for every $X\in\CC$ because it is the kernel of the ultraweakly continuous map 
$$\Hom_{\LMod_A}(X,X) \to \Hom_{\LMod_A}(T(X),X), \quad f\mapsto t_X\circ T(f)-f\circ t_X.$$
In particular, we obtain a von Neumann algebra
$$B=\Hom_\CC(R,R)^\op$$
so that $R$ is an $A$-$B$-bimodule. Moreover, the right action of $A$ on $R$ induces a homomorphism $A\to B$.

We claim that the functor 
$$R\boxtimes_B-: \LMod_B \to \CC$$
is an equivalence. Indeed, we have $\Hom_\CC(R,R) \cong \Hom_{\LMod_B}(L^2(B),L^2(B))$, thus the functor is fully faithful by Lemma \ref{lem:lmod-l2}. Moreover, since every object $X\in\CC$ is a direct summand of $T(X)=R\boxtimes_A X$ which is a direct summand of a direct sum of copies of $R$ by Lemma \ref{lem:lmod-l2}, the functor is essentially surjective.

We have a pair of completely additive functors
$$M:\CC\to\LMod_A, \quad X\mapsto X,$$
$$N:\LMod_A\to\CC, \quad X\mapsto T(X).$$
Note that $M\circ N = T$.
Let $t:N\circ M\to\Id_\CC$ be the natural transformation given by the action map $t_X:T(X)\to X$ so that $t\circ t^*=\Id_{\Id_\CC}$.
Regarding $M$ and $N$ as bimodules, we conclude the theorem.
\end{proof}

\begin{rem} \label{rem:vn-cc}
From the proof we see that the $*$-condensation $A\condense B$ is canonically determined by the $*$-condensation monad $R$: the von Neumann algebra $B^\op$ is the kernel of the ultraweakly continuous map
$$\Hom_{\LMod_A}(R,R) \to \Hom_{\LMod_A}(R\boxtimes_A R,R), \quad f\mapsto m\circ(R\boxtimes_A f)-f\circ m$$
and we have $M={}_A R_B$ and $N={}_B L^2(B)_A$.
\end{rem}

\begin{rem} \label{rem:vn-cc-dual}
If the $*$-condensation monad $R$ is unital, then $T$ is a monad over $\LMod_A$ so that $t:N\circ M\to\Id_\CC$ is a counit map exhibiting the forgetful functor $M$ right adjoint to $N$. 
According to the next remark, $L^2(B)\cong R$ canonically in this case.
\end{rem}

\begin{rem} \label{rem:vnbim-dual}
If an $A$-$B$-bimodule $M$ is dualizable (as a 1-morphism of $\Mor(\VN)$) then it is dual to the $B$-$A$-bimodule $\bar M$ given by the complex conjugate. 
This is \cite[Corollary 6.12]{BDH14} if $Z(A)$ is finite-dimensional. To see the general case, we may assume that $M$ is a faithful $A$-$B$-bimodule. Then \cite[Proposition 3.10]{BDH14} reduces the problem to the case where $A\subset B$ and $M=L^2(B)$.
Suppose an $A$-$A$-bimodule map $u:L^2(A)\to M\boxtimes_B N$ exhibits $N$ left dual to $M$. Then $u^*u\in Z(A)$ is nonvanishing almost everywhere. Replacing $u$ by $u(u^*u)^{-1/2}$, we assume $u^*u=1$. Then the proof of \cite[Proposition 7.3]{BDH14} supplies a conditional expectation $E:B\to A$ such that $E(x)\ge x$ for $x\in B_+$. The proof of \cite[Proposition 7.5]{BDH14} then shows that $M$ is dual to $\bar M$.
\end{rem}

\end{document}